\documentclass{lmcs}
\usepackage[utf8]{inputenc}
\pdfoutput=1

\usepackage{lastpage}
\lmcsdoi{18}{1}{33}
\lmcsheading{}{\pageref{LastPage}}{}{}%
{Feb.~24,~2021}{Feb.~24,~2022}{}

\usepackage{lineno}
\usepackage{stmaryrd}
\usepackage{amssymb}
\usepackage{color}
\usepackage{xcolor}
\usepackage{array}
\usepackage{url}
\usepackage{listings}
\usepackage{ebproof}
\usepackage{amsmath}
\usepackage{subcaption}
\usepackage{thmtools}
\usepackage{thm-restate}

\newcommand{\ens}[1]{\mathbb{#1}}
\newcommand{\expr}[1]{{\tt #1}}
\newcommand{\e}{\expr{e}}
\newcommand{\Operator}{\ens{O}}
\newcommand{\Var}{\ens{V}}
\newcommand{\oper}[1]{{\tt #1}}
\newcommand{\op}{\oper{op}}
\newcommand{\variable}[1]{{\tt #1}}
\newcommand{\vvv}{\variable{v}}
\newcommand{\uu}{\variable{u}}
\newcommand{\ww}{\variable{w}}
\newcommand{\x}{\variable{x}}
\newcommand{\y}{\variable{y}}
\newcommand{\z}{\variable{z}}

\newcommand{\Command}[1]{{\tt #1}}
\newcommand{\cmd}{\Command{c}}
\newcommand{\prog}{{\tt p}_\phi}

\newcommand{\instr}[1]{{\tt #1}}
\newcommand{\skp}{\instr{skip}}
\newcommand{\sap}{\instr{;\ }}
\newcommand{\asg}{\instr{\ :=\ }} 
\newcommand{\while}{\instr{while}}
\newcommand{\ret}{\instr{return \ }}
\newcommand{\ifa}{\instr{if}}

\newcommand{\elsea}{\instr{\ else\ }}
\newcommand{\rgl}{::=}
\newcommand{\FV}{\mathcal{V}}
\newcommand{\E}{\mathcal{E}}
\newcommand{\Cb}{\mathcal{C}}
\newcommand{\D}{\mathcal{D}}
\newcommand{\A}{\mathcal{A}}
\newcommand{\Oo}{Op}
\newcommand{\true}{1}
\newcommand{\false}{0}
\newcommand{\W}{\ens{W}}
\newcommand{\w}{\mathit{w}}
\newcommand{\varv}{\mathit{v}}
\newcommand{\varu}{\mathit{u}}
\newcommand{\size}[1]{|#1|}
\newcommand{\sem}[1]{\llbracket#1\rrbracket}
\newcommand{\store}{\mu}
\newcommand{\Imp}{\vDash}
\newcommand{\tier}[1]{\mathbf{#1}}
\newcommand{\tiera}{\tier{0}}
\newcommand{\tierb}{\tier{1}}
\newcommand{\tierc}{\tier{2}}
\newcommand{\tierd}{\tier{3}}
\newcommand{\slat}[1]{\mathbf{#1}}
\newcommand{\sla}{\slat{t}}
\newcommand{\slb}{\slat{t'}}
\newcommand{\SL}{\mathbf{N}}
\newcommand{\meet}{\wedge}
\newcommand{\join}{\vee}
\newcommand{\ord}{\preceq}
\newcommand{\ordst}{\prec}
\newcommand{\pbl}{\Gamma, \Delta \vdash} 
\newcommand{\dord}{\unlhd}
\newcommand{\typenv}{\Gamma}
\newcommand{\typop}{\Delta}
\newcommand{\dom}{\textit{dom}}
\newcommand{\meq}[1]{==_{#1}}
\newcommand{\mpred}{\oper{pred}}
\newcommand{\msuc}[1]{\oper{suc}_{#1}}

\newcommand{\mpt}{\mathrm{MPT}}

\newcommand{\st}{\mathrm{ST}}
\newcommand{\FP}{\mathrm{FP}}
\newcommand{\BFF}{\mathrm{BFF}}


\def\eg{\emph{e.g.}}
\def\ie{\emph{i.e.}}

\begin{document}

\title[A tier-based typed PL characterizing BFF]{A tier-based typed programming language characterizing Feasible Functionals}

\author[E.~Hainry]{Emmanuel Hainry\rsuper{a}}
\author[B.~Kapron]{Bruce M. Kapron\rsuper{b}}
\author[J-Y.~Marion]{Jean-Yves Marion\rsuper{a}}
\author[R.~P\'echoux]{Romain P\'echoux\rsuper{a}}

\address{Universit{\'e} de Lorraine, CNRS, Inria, LORIA, F-54000 Nancy, France}
\email{\{emmanuel.hainry,jean-yves.marion,romain.pechoux\}@loria.fr}
\address{University of Victoria, Victoria, BC, Canada}
\email{bmkapron@uvic.ca}

\begin{abstract}
The class of Basic Feasible Functionals $\BFF_2$ is the type-2 counterpart of the class $\FP$ of type-1 functions computable in polynomial time.
Several characterizations have been suggested in the literature, but none of these present a programming language with a type system guaranteeing this complexity bound.
We give a characterization of $\BFF_2$ based on an imperative language with oracle calls using a tier-based type system whose inference is decidable. Such a characterization should make it possible to link higher-order complexity with programming theory.
The low complexity (cubic in the size of the program) of the type inference algorithm contrasts with the intractability of the aforementioned methods and does not overly constrain the expressive power of the language.
\end{abstract}

\maketitle


\section{Introduction}
Type-2 computational complexity aims to study classes of functions that take type-1 arguments. The notion of feasibility for type-2 functionals was first studied in~\cite{Con73} and in~\cite{M76} using subrecursive formalisms. Later,~\cite{CooKap89,CU93} provided characterizations of polynomial time complexity at all finite types based on programming languages with explicit bounds and applied typed lambda-calculi, respectively. The class characterized in these works was christened the Basic Feasible Functionals, $\BFF$ for short.

It was shown in~\cite{KC91,KC96} that, similarly to type-1, feasible type-2 functions correspond to the programs computed in time polynomial in the size of their input. In this setting, the polynomial bound is a type-2 function as the size of a type-1 input is itself a type-1 object. This characterization lent support to the notion that at type level 2, the Basic Feasible Functionals ($\BFF_2$) are the correct generalization of FP to type-2.

Nevertheless, these characterizations are faced by at least two problems:
\begin{enumerate}
\item Characterizations using a general model of computation (whether machine- or program-based) require externally imposed and explicit resource bounding, either by a type-2 polynomial~\cite{KC91,KC96,FHHP15} or a bounding function within the class of~\cite{Con73,M76}. This is analogous to a shortcoming in Cobham's characterization of the class of (type 1) polynomial time computable functions $\FP$~\cite{Cob65}. Such bounding requires either a prior knowledge of program complexity or a check on type-2 polynomial time constraints, which is  highly intractable;\label{pb:1}
\item There is no natural programming language for these characterizations as they rely on machines or function algebras and cannot be adapted directly to programs. Some attempts have been made to provide programming languages for characterizing $\BFF_2$.  These languages are problematic either due to a need to provide some form of explicit external bounding~\cite{CooKap89} or from including unnatural constructs or type-2 recursion patterns~\cite{CU93,IRK01,DR06} which severely constrain the way in which type-2 programs may be written. All these distinct approaches would make it difficult for a non-expert programmer to use these formalisms as programming languages.\label{pb:2}
\end{enumerate}

\noindent
A solution to Problem~(\ref{pb:1}) was suggested in~\cite{KS17} by constraining Cook's definition of Oracle Polynomial Time (OPT)~\cite{C92}, which allows type-1 polynomials to be substituted for type-2 polynomials. To achieve this, oracle Turing machines are required to have a \emph{polynomial step count}: on any input, the length of their computations is bounded by a type-1 polynomial in the size of their input and the maximal size of any answer returned by the oracle. However $\BFF_2$ is known to be strictly included in OPT\@. In~\cite{KS17}, OPT is constrained by only allowing computations in which oracle return values increase in size a constant number of times, resulting in a class they called SPT (\emph{strong polynomial time}). This class is strictly contained in $\BFF_2$. $\BFF_2$ is recovered in~\cite{KS18} by putting a dual restriction, called \emph{finite lookahead revision}, on machines: on any input, the number of oracle calls on input of increasing size is bounded by a constant. The class of functions computed by machines having polynomial step count and finite lookahead revision is called MPT\@. The type-2 restriction of the simply-typed lambda closure of functions in MPT (and SPT) characterizes exactly $\BFF_2$.

Problem~(\ref{pb:2}) has been extensively tackled by the Implicit Computational Complexity community for type-1 complexity. This line of work provides machine independent characterizations that eliminate the external explicit bound and was initiated by the seminal works~\cite{BelCoo92} and~\cite{LeivantMar93}.
However, none of these works has been adapted to the case of type-2 complexity in a tractable approach. To this day, tractable implicit characterizations of type-2 complexity classes are still missing.
\paragraph{Our contribution}
We provide the first tractable characterization of type-2 polynomial time using a typed imperative language with oracle calls. Each oracle call comes with an associated \emph{input bound} which aims at bounding the size of the oracle input. However the size of the oracle answer, which is unpredictable, remains unbounded and, consequently, the language can be used in practice.

The characterization is inspired by the tier-based type system of~\cite{M11} characterizing $\FP$. Consequently, it relies on a non-interference principle and is also inspired by the type system of~\cite{VIS96} guaranteeing confidentiality and integrity policies by ensuring that values of high level variables do not depend on values of low level variables during a program execution. In our context, the level is called a tier.

Let $\sem{\st}$ be the set of functions computed by typable (also called \emph{safe}, see Definition~\ref{def:safe}) and terminating programs and let $\lambda(X)_2$ be the type-2 restriction of the simply-typed lambda closure of terms with constants in $X$. The  characterization of $\BFF_2$ is as follows:

\begin{thm}\label{thm:main}
$\lambda(\sem{\st})_2 = \BFF_2$.
\end{thm}

Soundness ($\lambda(\sem{\st})_2 \subseteq \BFF_2$, Theorem~\ref{thm:soundness}) is demonstrated by showing that each function of $\sem{\st}$ is in Kapron-Steinberg's MPT class~\cite{KS18}. The type system makes use of several tiers and is designed to enforce  a  tier-based non-interference result (Theorem~\ref{thm:ni}) and generalizes the operator type discipline of~\cite{M11} to ensure the polynomial step count property (Corollary~\ref{psc}) and the finite lookahead revision property (Theorem~\ref{thm:flr}), two non-trivial semantic properties. Two important points to stress are that: (i) these properties are enforced statically on programs as consequences of being typable (whereas they were introduced in~\cite{KS18} as pure semantic requirements on machines); (ii) the enforcement of finite lookahead revision through the use of tiering is a new non-trivial result.

Completeness ($  \BFF_2 \subseteq \lambda(\sem{\st})_2$, Theorem~\ref{t2}) is shown using an alternative characterization: $\lambda(\FP \cup \{\mathcal{I}'\})_2 = \BFF_2$, where $\mathcal{I}'$ is a bounded iterator that is polynomially equivalent to the recursor $\mathcal{R}$ of~\cite{CU93}, as demonstrated in~\cite{KS19}. The simulation of $\FP$ is performed by showing that our type system strictly embeds the tier-based type system of~\cite{MP14}. Consequently, our type system also provides a characterization of $\FP$ (Theorem~\ref{t1}) with strictly more expressive power when restricted to type-1 programs. Finally, a typable and terminating program computing the bounded iterator functional $\mathcal{I}'$ is exhibited. As in~\cite{KS18}, the simply-typed lambda-closure is mandatory to achieve completeness as oracle composition is not allowed by the syntax of the language.

The tractability of the type system is proved in Theorem~\ref{thm:ti}, where type inference is shown to be to be solvable in cubic time in the size of the program.
As a consequence of the decidability of type inference for simply typed lambda-calculus~\cite{M91}, we obtain the first decidable (up to a termination assumption) programming language based characterization of type-2 polynomial.
While the termination assumption is obviously not decidable, it is the most general condition for the result to hold.
However, it can be replaced without loss of completeness by combining our type system with automatic termination provers for imperative programs, for example~\cite{CPR06,LJB01}. The price to pay is a loss of expressive power.
Hence this paper provides a new approach for reasoning about type-2 feasibility automatically, in contrast to related works.

The characterization of Theorem~\ref{thm:main} is extensionally complete: all functions of $\BFF_2$ are computed by a typable and terminating program. It is not intensionally complete: there are false negatives as discussed in Example~\ref{oracles}. This incompleteness is a consequence of the decidability of type inference as providing intensionally complete descriptions of polynomial time is known to be a $\Sigma_2^0$-complete problem in the arithmetical hierarchy~\cite{H79}.

\textit{Outline. }\S\ref{s:ts} is devoted to presenting the type system technical developments and main intuitions. \S\ref{s:prop} states the type system main properties. \S\ref{s:ex} presents several examples that will help the reader to understand the underlying subtle mechanisms. Soundness and completeness are proved in \S\ref{s:sound} and \S\ref{s:comp}, respectively. The decidability of type inference is shown in \S\ref{s:ext}.
Future work is discussed in \S\ref{s:con}.

This paper is an extended and improved version of the paper~\cite{HKMP20} presented at Logic In Computer Science 2020, including complete proofs.

\section{Related work}
\paragraph{Implicit Computational Complexity (ICC)} has lead to the development of several techniques such as interpretations~\cite{BMM11}, light logics~\cite{Girard98}, mwp-bounds~\cite{BAJK08,JK09}, and tiering~\cite{M11,LM13,HP15}. These tools are restricted to type-1 complexity. Whereas the light logic approach can deal with programs at higher types, its applications are restricted to type-1 complexity classes such as $\FP$~\cite{BT04,BM10} or polynomial space~\cite{GMR08}. Interpretations were extended to higher-order polynomials in~\cite{BL16} to study $\FP$ and adapted in~\cite{FHHP15,HP17} to $\BFF_2$. However, by essence, all these characterizations use (at least) type-2 polynomials and cannot be considered as tractable.

\paragraph{Other characterizations of $\BFF_2$}
The characterizations of~\cite{CooKap89,IRK01} are based on a simple imperative programming language that enforces an explicit  external bound on the size of oracle outputs within loops. This restriction is impractical from a programming perspective as the size of oracle outputs cannot be predicted. In this paper, the bound is programmer friendly by its implicit nature and because it only constraints the size of the oracle input. Function algebra characterizations were developed in~\cite{KS19,CU93}: the recursion schemes are not natural and cannot be used in practice. Several characterizations~\cite{KC91,KC96} using type-2 polynomials were also developed but they focus on machines rather than programs.

\section{Imperative programming language with oracles}
\subsection{Syntax and semantics}
Consider a set $\Var$ of variables and a set $\Operator$ of operators $\op$ of fixed arity $ar(\op)$. For notational convenience, operators are used both in infix and prefix notations.  Let $\overline{t}$ denote a tuple of $n$ elements (variables, expressions, words, ...) $t_1,\ldots,t_n$, where $n$ is given by the context. 

Expressions, commands and programs are defined by the grammar of Figure~\ref{fig:synt},
\begin{figure*}
\hrulefill
\\[10pt]
$ \begin{array}{llll}
\texttt{Expressions} \qquad \qquad & \e, \e_1,\ldots
  & \rgl &\x
   \ |\ \op(\overline{\e})  \ | \ \phi(\e_1 \upharpoonright \e_2) \\
\texttt{Commands} &  \cmd, \cmd_1, \cmd_2
    & \rgl &\skp   \ | \  \x  \asg \e   \ | \  \cmd_1 \sap \cmd_2 \ | \ \ifa (\e) \{\cmd_1\} \elsea \{\cmd_2\}   \\
    & & &\ |\    \while(\e)\{ \cmd \}\   \\
  \texttt{Programs} &  \prog   & \rgl &\cmd \ \ret \x
 \end{array}$
 \\[10pt]

\hrulefill
\caption{Syntax of imperative programs with oracles}%
\label{fig:synt}
\end{figure*}
where $\x,\y \in \Var, \op, \upharpoonright \in \Operator$,
and $\phi$ is an oracle symbol. There can be only one oracle per program. Consequently, each program is indexed by its oracle as subscript.\footnote{The results can be generalised naturally to a constant number of oracles. However, this is of no particular interest with respect to the complexity class $\BFF_2$.}

Let $\FV(\prog)$ be the set of variables occurring in the program $\prog$. An expression of the shape $\phi(\e_1 \upharpoonright \e_2)$ is called an \emph{oracle call}. $\e_1$ is called the \emph{input data}, $\e_2$ is called the \emph{input bound} and $\e_1 \upharpoonright \e_2$ is called the \emph{input}. We write $\phi \notin \prog$ in the special case where no oracle call appears in $\prog$.

Let $\W = \Sigma^*$ be the set of words over a finite alphabet  $\Sigma$ such that $\{0,1\} \subseteq \Sigma$. The symbol $\epsilon$ denotes the empty word. The length of a word $\w$ (tuple $\overline{t}$) is denoted  $\size{\w}$ ($\size{\overline{t}}$, respectively). Given two words $\w$ and $\varv$ in $\W$ let $\varv.\w$ denote the concatenation of $\varv$ and $\w$. For a given symbol $a \in \Sigma$, let $a^n$ be defined inductively by $a^0 = \epsilon$ and $a^{n+1}=a.a^n$. Let $\dord$ be the sub-word relation over $\W$, which is defined by $\varv \dord \w$, if there are $\varu$ and $\varu'$ such that $\w = \varu.\varv.\varu'$.

A total function $\sem{\op}: \W^{ar(\op)} \to \W$ is associated to each operator. Constants may be viewed as operators of arity zero.

For a given word $\w \in \W$ and an integer $n$, let $\w_{\upharpoonright n}$ be the word obtained by truncating $\w$ to its first $\min(n, \size{w})$ symbols and then padding with a word of the form $10^k$ to obtain a word of size exactly $n+1$. For example, $1001_{\upharpoonright 0} = 1$, $1001_{\upharpoonright 1} = 11$, $1001_{\upharpoonright 2} = 101$, and $1001_{\upharpoonright 6} = 1001100$.
Define $\forall \varv,\w \in \W,\ \sem{\upharpoonright}(v,w) =  v_{\upharpoonright \size{w}}$.
Padding ensures that $\size{\sem{\upharpoonright}(v,w)}=\size{\w}+1$. The syntax of programs enforces that oracle calls are always performed on input data padded by the input bound. Combined with the above property, this ensures that oracle calls are always performed on input data whose size does not exceed the size of the input bound plus one. Consequently, no oracle call can be performed on the empty word.

The oracle symbol $\phi$ computes a total function from $\W$ to $\W$, called an oracle function. In order to lighten notations, we will make no distinction between the oracle symbol $\phi$ and the oracle function it represents.

A store $\store$ is a partial map from $\Var$ to $\W$. Let $dom(\store)$ be the domain of $\store$. Let $\store[\x_1 \leftarrow \w_1,\ldots,\x_n \leftarrow \w_n]$ be a notation for the store $\store'$ satisfying $\forall \x \in \dom(\store)-\{\x_1,\ldots,\x_n\},\ \store'(\x)=\store(\x)$ and $\forall \x_i \in\{\x_1,\ldots,\x_n\},\ \store'(\x_i)=\w_i$. Let $\store_0$ be the store defined by $dom(\store_0)=\Var$ and $\forall \x \in dom(\store_0),\ \store_0(\x)=\epsilon$.
The size of a store $\store$ is defined by $\size{\store}= \Sigma_{\x \in \dom(\store)}\size{\store(\x)}.$

The judgment $\store \Imp_\phi \e \to  \w$ means that the expression $\e$ is evaluated to the word $\w \in \W$ with respect to the store $\store$ and the oracle $\phi$.  The judgment $\store \Imp_\phi \cmd \to \store'$ expresses that, under the store $\store$ and the oracle $\phi$, the command $\cmd$ terminates and outputs the store $\store'$. As the oracle is fixed for each program, we will omit it throughout the paper in the judgments subscript, \eg, writing $\store \Imp \e \to  \w$ for $\store \Imp_\phi \e \to  \w$. The operational semantics of the language is deterministic and is given in Figure~\ref{fig:Com}. In rule (Seq) of Figure~\ref{fig:Com}, it is implicitly assumed that $\cmd_1$ is not a sequence.

A \emph{derivation} $\pi_{\phi}: \store \Imp \prog \to \w$ is a tree rooted at $ \store \Imp \prog \to \w$, where children of each node are obtained by applying the rules of Figure~\ref{fig:Com}.
Let $\size{\pi_{\phi}}$ denote the size of the derivation $\pi_{\phi}$. Note that $\size{\pi_{\phi}}$ corresponds to the number of steps in a sequential execution of $\prog$, initialized with store $\mu$. Hence, $\size{\pi_{\phi}}$ can be infinite. With no restriction on operators, this measure is too coarse to correspond, even asymptotically, to running time. With suitable restrictions, there is a correspondence, given in Proposition~\ref{runtime} below.

A program $\prog$ such that $\FV(\prog)=\{\overline{\x}\}$ computes the partial function $\sem{\prog} \in \W^{\size{\overline{\x}}} \to \W$,
defined by $\sem{\prog}(\overline{\w})=w$ if $\exists \pi_{\phi}, \ \pi_{\phi}: \store_0[\x_1 \leftarrow \w_1,\ldots,\x_{\size{\overline{\x}}} \leftarrow \w_{\size{\overline{\x}}}] \Imp \prog \to  \w.$
In the special case where, for any oracle $\phi$, $\sem{\prog}$ is a total function, the program $\prog$ is said to be terminating.

A second order function $f: (\W \to \W) \to (\W \to \W)$ is computed by a program $\prog$ if for any oracle function $\phi \in \W \to \W$ and word $\w \in \W$, $f(\phi)(\w)=\sem{\prog}(\w)$.

\begin{figure*}[t]
\hrulefill
\\[10pt]
\centering
\begin{prooftree}
\hypo{\phantom{ \Imp \x  \asg \e \to } }
\infer1[(Var)]{\store \Imp \x \to \store(\x)}
\end{prooftree}
\quad \quad
\begin{prooftree}
\hypo{\forall i \leq ar(\op),\ \store \Imp \e_i \to \w_i}
\infer1[(Op)]{\store \Imp \op(\overline{\e}) \to \sem{\op}(\overline{\w})}
\end{prooftree}
\\[10pt]
\begin{prooftree}
\hypo{\store \Imp \e_1 \to \varv}
\hypo{\store \Imp \e_2 \to \w}
\hypo{\phi(\sem{\upharpoonright}(\varv,\w))=\varu}
\infer3[(Orc)]{\store \Imp \phi(\e_1\upharpoonright \e_2) \to \varu}
\end{prooftree}
\\[10pt]
\begin{prooftree}
\hypo{\phantom{ \Imp \x  \asg \e \to }  }
\infer1[(Skip)]{\store \Imp \skp \to \store}
\end{prooftree}
\quad \quad
\begin{prooftree}
\hypo{\store \Imp \e \to \w }
\infer1[(Asg)]{\store \Imp \x  \asg \e \to \store[\x  \leftarrow \w]}
\end{prooftree}
\\[10pt]
\begin{prooftree}
\hypo{\store \Imp \cmd_1 \to \store_1 }
\hypo{\store_1 \Imp \cmd_2 \to \store_2}
\infer2[(Seq)]{\store \Imp \cmd_1 \sap \cmd_2 \to \store_2}
\end{prooftree}
\quad \quad
\begin{prooftree}
\hypo{\store \Imp \e \to \w}
\hypo{\store \Imp \cmd_\w \to \store'}
\hypo{\w\in \{\false,\true\} }
\infer3[(Cond)]{\store \Imp \ifa (\e) \{\cmd_{\true}\} \elsea \{\cmd_{\false}\}  \to \store'}
\end{prooftree}
\\[10pt]
\begin{prooftree}
\hypo{\store \Imp \e \to \false}
\infer1[(Wh$_0$)]{\store \Imp \while (\e) \{\cmd\} \to \store}
\end{prooftree}
\quad \quad
\begin{prooftree}
\hypo{\store \Imp \e \to \true}
\hypo{\store \Imp \cmd \sap  \while (\e) \{\cmd\} \to \store'}
\infer2[(Wh$_1$)]{\store \Imp  \while (\e) \{\cmd\} \to \store'}
\end{prooftree}
\\[10pt]
\begin{prooftree}
 \hypo{\store \Imp \cmd \to \store'}
 \infer1[(Prg)]{\store \Imp  \cmd\ \ret \x \to \store'(\x)}
\end{prooftree}
\\[10pt]
\hrulefill
\caption{Big step operational semantics}%
\label{fig:Com}
\end{figure*}

\subsection{Neutral and positive operators}
We define two classes of operators called neutral and positive. This categorization of operators will be used in \S\ref{ss:safe} where the admissible types for operators will depend on their category in the type system.

\begin{defi}[Neutral and positive operators]\label{def:np}\quad \quad
\begin{itemize}
\item An operator $\op$ is \emph{neutral} if:
\begin{enumerate}
\item either $\op$ is a constant operator, \ie, $ar(\op)=0$;
\item $\sem{\op}:\W^{ar(\op)} \to \{0,1\}$, \ie, $\sem{\op}$ is a predicate;
\item or  $\forall \overline{\w} \in \W^{ar(\op)},\ \exists i  \leq ar(\op)$, $
\sem{\op}(\overline{\w}) \dord {\w_i}
$;
\end{enumerate}
\item An operator $\op$  is \emph{positive} if there is a constant $c_{\op}$ such that: \\$
\forall \overline{\w}^{ar(\op)} \in \W,\ \size{\sem{\op}(\overline{\w})}  \leq \max_i \size{\w_i} + c_{\op}
$.
\end{itemize}
\end{defi}

\noindent
A neutral operator is always a positive operator but the converse is not true.
In the remainder, we name positive operators those operators that are positive but not neutral.

\begin{exa}\label{ex:1}
The operator $\meq{}$ tests whether or not its arguments are equal and the operator $\mpred{}$ computes the predecessor.
\[
\sem{\meq{}}(\w, \varv) =
   \begin{cases}
     1 &\text{if }\varv=\w  \\
         0 &\text{otherwise}
   \end{cases} \quad \quad
\sem{\mpred}(\varv)   =
     \begin{cases}
       \epsilon  &\text{if\ }\varv=\epsilon \\
      \varu  &\text{if\ } \varv=a.\varu,\ a \in \Sigma \\
      \end{cases}
\]
Both operators are neutral.
$
\sem{\msuc{i}}(\varv) = i.\varv, \ \text{for }i \in \{0,1\}$, is a positive operator since $\size{\sem{\msuc{i}}(\varv)}=\size{i.\varv}=\size{\varv}+1$.
\end{exa}

\section{Type system}\label{s:ts}
In this section, we introduce a tier based type system, the main contribution of the paper, that allows to provide a characterization of type-2 polynomial time complexity (\cite{M76,KC91,KC96}).

\subsection{Tiers and typing judgments}
Atomic types are elements of the totally ordered set $(\SL ,\ord,\tiera,\join,\meet)$ where $\SL =\{\tiera, \tierb, \tierc, \ldots\}$ is the set of natural numbers, called \emph{tiers}, in accordance with the data ramification principle of~\cite{L94},  $\ord$ is the usual ordering on integers and $\join$ and $\meet$ are the max and min operators over integers. Let $\ordst$ be defined by $\ordst \ :=\ \ord \cap \neq$.
We use the symbols $\sla,\slb,\ldots,\sla_1,\sla_2,\ldots$ to denote tier variables. For a finite set of tiers, $\{\sla_1,\ldots,\sla_n\}$, let $\join_{i=1}^n \sla_i$ ($\meet_{i=1}^n \sla_i$, respectively) denote $\sla_1 \join \ldots \join \sla_n$ ($\sla_1 \meet \ldots \meet \sla_n$, respectively).

 A variable typing environment $\typenv$ is a finite mapping from $\Var$ to $\SL$, which assigns a single tier to each variable.

An operator typing environment $\typop$ is a mapping that associates to each operator $\op$ and each tier $\sla \in \SL$ a set of admissible operator types $\typop(\op)(\sla)$, where the operator types corresponding to the operator $\op$ are of the shape $\sla_1 \to \ldots \to \sla_{ar(\op)} \to \sla'$, with $\sla_i,\sla' \in \SL$.

Let $\dom(\typenv)$ (resp. $\dom(\typop)$)  denote the set of variables typed by $\typenv$ (resp.\ operators typed by $\typop$).

Typing judgments are either \emph{command typing judgments} of the shape $\pbl \cmd : (\sla,\sla_{in},\sla_{out})$ or \emph{expression typing judgments} of the shape $\pbl \e : (\sla,\sla_{in},\sla_{out})$. The intended meaning of such a typing judgment is that the \emph{expression tier} or \emph{command tier} is $\sla$, the \emph{innermost  tier} is $\sla_{in}$, and the \emph{outermost tier} is $\sla_{out}$.
The innermost (resp.\ outermost) tier is the tier of the guard of the innermost (resp.\ outermost) while loop containing the expression or command in question.
In the case of a single non-nested while loop, the innermost and outermost tiers are equal (as illustrated by rule (W$_0$) of Figure~\ref{TS}). These two tiers are irrelevant for an expression or a command not appearing inside a while loop.

The type system preventing flows from $\sla_2$ to $\sla_1$, whenever $\sla_2 \ordst \sla_1$ holds, is presented in Figure~\ref{TS}.

A typing derivation $\rho \rightslice \pbl \cmd :(\sla,\sla_{in},\sla_{out})$  is a tree whose root is the typing judgment $\pbl \cmd :(\sla,\sla_{in},\sla_{out})$ and whose children are obtained by applications of the typing rules. Due to the rule (OP) of Figure~\ref{TS}, that allows several admissible types for operators, typing derivations are, in general, not unique.
However the two typing rules for while loops (W) and (W$_0$) are mutually exclusive (when read bottom-up) because of the non-overlapping requirements for $\sla_{out}$ in Figure~\ref{TS}.
 The notation $\rho$ will be used whenever mentioning the root of a typing derivation is not explicitly needed. We use the notation $\rho \rightslice \pbl \cmd :(\sla,\sla_{in},\sla_{out})$ (R) to denote the typing derivation $\rho \rightslice \pbl \cmd :(\sla,\sla_{in},\sla_{out})$ whose children are obtained by application of a typing rule labelled by (R).

 Given two typing derivations $\rho$ and $\rho'$, we write $\rho' \leq \rho$ (respectively $\rho' < \rho$) if $\rho'$ is a (strict) subtree of $\rho$.
Let $\D(\rho)$ be  defined by $\D(\rho)=\{\rho' \ | \ \rho' \leq \rho\}$ and let $\mathring{\D} (\rho)$ be defined by $\mathring{\D} (\rho)=\D(\rho)-\{ \rho\}$.

\begin{figure*}
\hrulefill
\centering
\\[10pt]
\centering
\begin{prooftree}
\hypo{\Gamma(\x)=\sla}
\infer1[(V)]{\pbl \x: (\sla, \sla_{in},\sla_{out})}
\end{prooftree}
\\[10pt]
{
\begin{prooftree}
\hypo{\sla_1 \to \cdots \to \sla_{ar(\op)} \to \sla \in \Delta(\op)(\sla_{in})}
\hypo{\forall i \leq ar(\op), \  \pbl \e_i: (\sla_i,\sla_{in},\sla_{out})}
\infer2[(OP)]{\pbl  \op(\overline{\e}): (\sla,\sla_{in},\sla_{out})}
\end{prooftree}
}
\\[10pt]
{
\begin{prooftree}
\hypo{\pbl \e_1:(\sla,\sla_{in},\sla_{out})}
\hypo{\pbl \e_2: (\sla_{out},\sla_{in},\sla_{out})}
\hypo{\sla \ordst  \sla_{in} \wedge \sla \ord \sla_{out}  }
\infer3[(OR)]{\pbl \phi(\e_1 \upharpoonright \e_2) : ( \sla,\sla_{in},\sla_{out})}
\end{prooftree}
}
\\[10pt]
\begin{prooftree}
\hypo{\pbl  \cmd\ : (\sla,\sla_{in},\sla_{out})}
\infer1[(SUB)]{\pbl  \cmd\ : (\sla\mathbf{+1},\sla_{in},\sla_{out})}
\end{prooftree}
\quad \quad
\begin{prooftree}
\infer0[(SK)]{\pbl \skp\ : (\tiera,\sla_{in},\sla_{out})}
\end{prooftree}
\\[10pt]
\begin{prooftree}
\hypo{\pbl \x: (\sla_1,\sla_{in},\sla_{out})}
\hypo{\pbl \e: (\sla_2,\sla_{in},\sla_{out})}
\hypo{\sla_1  \ord \sla_2}
\infer3[(A)]{\pbl \x \asg \e \ : (\sla_1,\sla_{in},\sla_{out})}
\end{prooftree}
\\[10pt]
\begin{prooftree}
\hypo{\pbl  \cmd_1 :(\sla,\sla_{in},\sla_{out})}
\hypo{\pbl \cmd_2:(\sla,\sla_{in},\sla_{out})}
\infer2[(S)]{\pbl \cmd_1 \sap \cmd_2 \ :(\sla,\sla_{in},\sla_{out})}
\end{prooftree}
\\[10pt]
{
\begin{prooftree}
\hypo{\pbl \e: (\sla,\sla_{in},\sla_{out})}
\hypo{\pbl \cmd_{\true}: (\sla,\sla_{in},\sla_{out})}
\hypo{\pbl \cmd_{\false}:(\sla,\sla_{in},\sla_{out})}
\infer3[(C)]{\pbl \ifa (\e) \{\cmd_{\true}\} \elsea \{\cmd_{\false}\}\ :(\sla,\sla_{in},\sla_{out})}
\end{prooftree}
}
\\[10pt]
\begin{prooftree}
\hypo{\pbl \e: (\sla,\sla_{in},\sla_{out})}
\hypo{\pbl \cmd : (\sla,\sla,\sla_{out})}
\hypo{\tierb \ord \sla \ord \sla_{out}}
\infer3[(W)]{\pbl  \while (\e) \{\cmd\}\ :(\sla,\sla_{in},\sla_{out})}
\end{prooftree}
\\[10pt]
\begin{prooftree}
\hypo{\pbl \e: (\sla,\sla_{in},\sla)}
\hypo{\pbl \cmd : (\sla,\sla,\sla)}
\hypo{\tierb \ord \sla }
\infer3[(W$_{0}$)]{\pbl  \while (\e) \{\cmd\}\ :(\sla,\sla_{in},\tiera)}
\end{prooftree}
\\[10pt]
\hrulefill
\caption{Tier-based type system}%
\label{TS}

\end{figure*}

\subsection{Safe environments and programs}\label{ss:safe}
The typing rules of Figure~\ref{TS} are not restrictive enough in themselves to guarantee polynomial time computation, even for type-1. Indeed operators need to be restricted to prevent exponential programs from being typable (see counter-Example~\ref{counter:exp}). The current subsection introduces such a restriction, called \emph{safe}.

\begin{defi}[Safe operator typing environment]\label{sote}
An operator typing environment $\typop$ is \emph{safe} if for each $\op \in \dom(\typop)$ of arity $ar(\op)>0$, $\op$ is neutral or positive and $\sem{\op}$ is a polynomial time computable function, and for each $\sla_{in} \in \SL$, and for each $\sla_1 \to \ldots \sla_{ar(\op)} \to \sla \in \typop(\op)(\sla_{in})$, the two conditions below hold:
\begin{enumerate}
\item $\sla \ord \meet_{i=1}^{ar(\op)} \sla_i \ord \join_{i=1}^{ar(\op)} \sla_i \ord \sla_{in}$,\label{premier}
\item  if the operator $\op$ is positive then $\sla \ordst \sla_{in}$.\label{second}
\end{enumerate}
\end{defi}

\begin{exa}\label{tre}
Consider the operators $\meq{}$, $\mpred$ and $\msuc{i}$ of Example~\ref{ex:1}.
For a safe typing environment $\Delta$, it holds that $\Delta(\meq{})(\tierb)=\{\tierb \to \tierb \to \tierb \} \cup \{\sla \to \sla' \to \tiera \ | \ \sla,\sla' \ord \tierb\}$, as $\meq{}$ is neutral. However $\tiera \to \tierb \to \tierb \notin \Delta(\meq{})(\tierb)$ as  it breaks Condition (\ref{premier}) of Definition~\ref{sote} since the operator output tier has to be smaller than each of its operand tier (\ie, $ \tierb \not\ord \tiera \meet  \tierb$).

It also holds that $\Delta(\mpred)(\tierc)=\{\tierc \to \sla \ | \ \sla \ord \tierc \}\cup \{ \tierb \to \sla \ | \ \sla \ord \tierb \} \cup \{\tiera \to \tiera \}$.

For the positive operator $\msuc{i}$, we have $\Delta(\msuc{i})(\tierb) =\{ \tierb \to \tiera ,\tiera \to \tiera\}$. $\tierb \to \tierb \notin \Delta(\msuc{i})(\tierb)$ as the operator output tier has to be strictly smaller than $\tierb$, due to Condition (\ref{second}) of Definition~\ref{sote}. Applying the same restriction, it holds that $\Delta(\msuc{i})(\tierc)=\{\tierc \to \tierb, \tierc \to \tiera, \tierb \to \tierb, \tierb \to \tiera, \tiera \to \tiera \}$.
\end{exa}

\begin{defi}[Safe program]~\label{def:safe}
Given $\typenv$ a variable typing environment and $\typop$ a safe operator typing environment, the program $\prog = \cmd  \ \ret \x$ is a \emph{safe program}  if there are $ \sla, \sla_{in},\sla_{out}$ such that $\rho \rightslice \pbl \cmd : (\sla,\sla_{in},\sla_{out})$.
\end{defi}

\begin{defi}
Let $\st$ be the set of safe and terminating programs and $\sem{\st}$ be the set of functionals computed by programs in $\st$:
\[\sem{\st}=\{ \lambda \phi.\lambda \w_1.\cdots \lambda w_n .\sem{\prog}(\w_1,\ldots,\w_n) \ | \ \prog \in \st\}.\]

\end{defi}

\subsection{Some intuitions}
Before providing a formal treatment of the type system's main properties in \S\ref{s:prop}, we provide the reader with a brief intuition of types, that are triplets of tiers $(\sla,\sla_{in},\sla_{out})$, in a typing derivation obtained by applying the typing rules of Figure~\ref{TS}:
\begin{itemize}
\item $\sla$ is the tier of the expression or command under consideration. It is used to prevent data flows from lower tiers to higher tiers in control flow commands and assignments. By safety and by rules (OP) and (OR), expression tiers are structurally decreasing. Consequently, rule (A) ensures that data can only flow from higher tiers to lower tiers. Command tiers are structurally increasing and, consequently, an assignment of a higher tier variable can never be controlled by a lower tier in a conditional or while command. The subtyping rule (SUB) for commands follows this discipline by allowing a command of tier $\sla$ to be considered as a tier $\sla\mathbf{+1}$ command and, hence, controlled by an expression of tier $\sla\mathbf{+1}$. However subtyping is strictly prohibited for expressions as this would break the flow.
\item $\sla_{in}$ is the tier of the innermost while loop containing the expression or command under consideration, provided it exists. It is used to allow declassification (\ie, a release of some information at a lower tier to a higher tier) to occur in the program by allowing an operator to have types depending on the context. Moreover, the innermost tier restricts the return types of operators and oracle calls:
\begin{itemize}
\item in rule (OR), the return type $\sla$ is strictly smaller than $\sla_{in}$,
\item in rule (OP), for a positive operator, the return type $\sla$ is strictly smaller than $\sla_{in}$.
\end{itemize}
This forbids programs from iterating on a data whose size can increase during the iteration.
\item $\sla_{out}$ is the tier of the outermost while loop containing the expression or command under consideration, provided it exists. Its purpose is to bound by a constant the number of lookahead revisions (that is the number of times a query to the oracle may increase in size) allowed in oracle calls. By rule (OR), all oracle input bounds have a tier equal to the tier of the outermost while loop where they are called. Hence, the size of the data stored in the input bound cannot increase in a fixed while loop and it can increase at most a constant number of times.
\end{itemize}
There are two rules (W) and (W$_0$) for while loops. (W) is the standard rule and updates the innermost tier with the tier of the while loop guard under consideration. (W$_0$) is an initialization rule that allows the programmer to instantiate by default the main command with outermost tier $\tiera$ as it has no outermost while. It could be sacrificed for simplicity but at the price of a worst expressive power.

\section{Examples}\label{s:ex}
In this section, we provide several examples and counter-examples, starting with programs with no oracle calls in order to illustrate how the type system works.
Some of its restrictions in terms of expressive power are also discussed in Example~\ref{oracles}. In the typing derivations, we sometimes omit the environments, writing $\vdash$ instead of $\pbl$ in order to lighten the notations. Moreover, for notational convenience, we will use labels for expression tiers. For example,  $\e^{\sla}$ means that $\e$ is of tier $\sla$. Also, to make the presentation of the examples lighter, we will work over the unary integers rather than all of $\W$. In particular, a value $\variable{v}$ denotes $1^\variable{v}$, and in particular $0$ denotes $\epsilon$. Also, with this convention, $\sem{\mpred}(\variable{v})=\max\{0,\variable{v}-1\}$ and $\sem{\msuc{1}}(\variable{v})=\variable{v}+1$.

\begin{exa}[Addition]\label{e1}
Consider the simple program below, with no oracle, computing the unary addition.

\begin{minipage}{0.9\textwidth}
\begin{lstlisting}[language=C, mathescape, frame=shadowbox, rulesepcolor=\color{blue!25},backgroundcolor=\color{blue!5}]
  $\while(\x > 0)^\tierb \{ $
    $\x^\tierb \asg \mpred(\x)^\tierb\sap$
    $\y^\tiera \asg \msuc{1}(\y)^\tiera$
  $\}$
  $\ret\ \y$
\end{lstlisting}
\end{minipage}

This program is safe with respect to the following typing derivation:
\[
\centering
\scalebox{0.95}{
\begin{prooftree}
\hypo{\Gamma(\x)=\tierb}
\infer1[(V)]{\vdash \x : (\tierb,\tierb,\tierb)}
\infer1[(OP)]{\vdash \x>0 : (\tierb,\tierb,\tierb)}
\hypo{}
\ellipsis{}{\rho_1 \rightslice  \vdash \x \asg \mpred(\x) : (\tierb,\tierb,\tierb) }
\hypo{}
\ellipsis{}{\rho_2 \rightslice \vdash \y \asg \msuc{1}(\y) : (\tiera,\tierb,\tierb)}
\infer2[(S)]{\vdash \x \asg \mpred(\x) \sap \y \asg \msuc{1}(\y) : (\tierb,\tierb,\tierb)}
\infer2[(W$_0$)]{\vdash \while(\x>0)\{\x \asg \mpred(\x) \sap \y \asg \msuc{1}(\y) \} : (\tierb,\tierb,\tiera)}
\end{prooftree}
}
 \]
The while loop is guarded by $\x > 0$. If the main command is typed by $({\tierb},\tierb, \tiera)$ then the expression $\x > 0$ is of tier $\tierb$ by the typing rule (W$_0$). Consequently, the variable $\x$ is forced to be of tier $\tierb$ using the type $\tierb \to \tierb$ for the operator $>0$ in the (OP) rule. $\tierb \to \tierb \in \Delta(> 0)(\tierb)$ holds as the operator $>0$ is neutral. One application of the subtyping rule (SUB) is performed for the sequence to be typed as the subcommands are required to have homogeneous types.

The typing derivation $\rho_1$ is as follows:
\[
\begin{prooftree}
\hypo{\Gamma(\x)=\tierb}
\infer1[(V)]{\vdash \x : (\tierb,\tierb,\tierb)}
\hypo{\Gamma(\x)=\tierb}
\infer1[(V)]{\vdash \x : (\tierb,\tierb,\tierb)}
\infer1[(OP)]{\vdash \mpred(\x) : (\tierb,\tierb,\tierb)}
\infer2[(A)]{\rho_1 \rightslice  \vdash \x \asg \mpred(\x) : (\tierb,\tierb,\tierb) }
\end{prooftree}
\]

In $\rho_1$, the $\mpred$ operator is used with the type $\tierb \to \tierb$ in the (OP) rule. This use is authorized as, $\mpred$ is neutral and, consequently, $\tierb \to \tierb \in \Delta(\mpred)(\tierb)$. As a consequence, the  rule (A) in $\rho_1$ can be derived as the tier of the assigned variable $\x$ (equal to $\tierb$) is smaller than the tier of the expression $\mpred(\x)$ (also equal to $\tierb$).

The second typing derivation $\rho_2$ is as follows:
\[
\begin{prooftree}
\hypo{\Gamma(\y)=\tiera}
\infer1[(V)]{\vdash \y : (\tiera,\tierb,\tierb)}
\hypo{\Gamma(\y)=\tiera}
\infer1[(V)]{\vdash \y : (\tiera,\tierb,\tierb)}
\infer1[(OP)]{\vdash \msuc{1}(\y) : (\tiera,\tierb,\tierb)}
\infer2[(A)]{\rho_2 \rightslice \vdash \y \asg \msuc{1}(\y) : (\tiera,\tierb,\tierb) }
\end{prooftree}
\]

The only distinction between $\rho_2$ and $\rho_1$ is that the operator $\msuc{1}$ is positive. Consequently, with an innermost tier of $\tierb$, the type $\tierb \to \tierb$ is not authorized for such an operator (since $\tierb \to \tierb \notin \Delta(\msuc{1})(\tierb)$). Indeed, by Example~\ref{tre}, $\Delta(\msuc{1})(\tierb) =\{ \tierb \to \tiera ,\tiera \to \tiera\}$. The type $\tierb \to \tiera$ is ruled out as it would require a non-homogeneous type for $\y$. Consequently, the rule (OP) is applied on type $\tiera \to \tiera$ and the variable $\y$ must be of tier $\tiera$. Notice that the program could also be typed by assigning higher tiers $\sla $ and $\sla'$ such that $\sla' \ordst \sla$, to $\x$ and $\y$, respectively.
\end{exa}

\begin{exa}[Exponential]\label{counter:exp}
The program below, computing the exponential, is not safe.

\begin{minipage}{0.9\linewidth}
\begin{lstlisting}
  $\while ( \x> 0)\{$
    $\z \asg \y\sap$
    $ \while (\z> 0)\{$
      $ \z \asg \mpred(\z) \sap$
      $\y \asg \msuc{1}(\y)$
    $\} \sap$
    $\x \asg \mpred(\x) $
  $\}$
  $\ret\ \y$
\end{lstlisting}
\end{minipage}

By contradiction, suppose that it can be typed with respect to the typing environments $\Gamma$ and $\Delta$. Let $\Gamma(\x),\Gamma(\y)$ and $\Gamma(\z)$ be $\sla_\x,\sla_\y$ and $\sla_\z$, respectively.

The subcommand $\z \asg \y$ enforces $\sla_{\z} \ord \sla_\y$ to be satisfied for the following typing derivation to hold.
\[
\begin{prooftree}
\hypo{\Gamma(\z)=\sla_\z}
\infer1[(V)]{\vdash \z : (\sla_\z,\sla_{in},\sla_{out})}
\hypo{\Gamma(\y)=\sla_\y}
\infer1[(V)]{\vdash \y : (\sla_\y,\sla_{in},\sla_{out})}
\infer2[(A)]{\rho_1 \rightslice \vdash \z \asg \y : (\sla_\z,\sla_{in},\sla_{out})}
\end{prooftree}
\]

The subcommand $\y \asg \msuc{1}(\y)$ enforces the constraint $\sla_{\y} \ordst \sla_{in}$, $\sla_{in}$ being the command innermost tier, for the typing derivation to hold.
\[
\begin{prooftree}
\hypo{\Gamma(\y)=\sla_\y}
\infer1[(V)]{\vdash \y : (\sla_\y,\sla_{in},\sla_{out})}
\hypo{\sla_\y \to \sla_\y \in \Delta(\msuc{1})(\sla_{in})}
\hypo{\Gamma(\y)=\sla_\y}
\infer1[(V)]{\vdash \y : (\sla_\y,\sla_{in},\sla_{out})}
\infer2[(OP)]{\vdash \msuc{1}(\y) : (\sla_\y,\sla_{in},\sla_{out})}
\infer2[(A)]{\rho_2 \rightslice\vdash \y \asg \msuc{1}(\y) : (\sla_\y,\sla_{in},\sla_{out}) }
\end{prooftree}
\]

Indeed, as $\msuc{1}$ is a positive operator, by Condition~\ref{second} of Definition~\ref{sote}, $\sla_{\y} \ordst \sla_{in}$ has to be satisfied for $\sla_\y \to \sla_\y \in \Delta(\msuc{1})(\sla_{in})$ to hold.

The innermost while loop enforces the constraint $ \sla_{in} \ord \sla_\z$ in the following typing derivation.
\[
\begin{prooftree}
\hypo{\Gamma(\z)=\sla_\z}
\infer1[(V)]{\vdash \z : (\sla_\z,\sla_{in}',\sla_{out})}
\infer1[(OP)]{\vdash \z>0 : (\sla_{in},\sla_{in}',\sla_{out})}
\hypo{}
\ellipsis{}{}
\infer1[(S)]{\vdash \z \asg \mpred(\z) \sap \y \asg \msuc{1}(\y) : (\sla_{in},\sla_{in},\sla_{out})}
\infer2[(W)]{\rho_3 \rightslice\vdash \while(\z>0)\{\z \asg \mpred(\z) \sap \y \asg \msuc{1}(\y) \} : (\sla_{in},\sla_{in}',\sla_{out})}
\end{prooftree}
\]

 First, notice that only the rule (W) can be applied to this typing derivation as the corresponding subcommand is already contained inside a while loop and, consequently, $\tierb \ord \sla_{out}$ is enforced by the outermost while loop using rule (W) or rule (W$_0$). Second, the tier of this subcommand is equal to the innermost tier $\sla_{in}$ of subcommand
$\y \asg \msuc{1}(\y)$ (in $\rho_2$). Indeed, rules (W) and (W$_0$) are the only typing rules updating the innermost tier and there is no while loop in between. Finally, in the rule (OP), as $>0$ is neutral, Condition~\ref{premier} of Definition~\ref{sote} enforces that $\sla_{in} \ord \sla_\z \ord \sla_{in}'$ holds for the program to be typed.

Putting all the above constraints together, we obtain the contradiction $\sla_\z  \ord \sla_\y \ordst \sla_{in} \ord \sla_\z$. Consequently, the program cannot be typed.
\end{exa}

%
%



%

\begin{exa}[Multiple tiers]\label{ex:multiple}
Consider the following program illustrating the use of multiple tiers.

\begin{lstlisting}
 $\cmd_1 : \while ( \x> 0)^\tierc\{$
      $\x^\tierc \asg \mpred(\x)^\tierc \sap$
      $\y^\tierb \asg \msuc{1}(\msuc{1}(\y))^\tierb$
    $\}$ $\sap$
 $\cmd_2 : \while ( \y> 0)^\tierb\{$
      $\y^\tierb \asg \mpred(\y)^\tierb \sap$
      $\z^\tiera \asg \msuc{1}(\msuc{1}(\z))^\tiera$
    $\}$
    $\ret\ \z$
\end{lstlisting}

The program is safe with respect to the variable typing environment $\Gamma$ such that $\Gamma(\x)=\tierc$, $\Gamma(\y)=\tierb$ and $\Gamma(\z)=\tiera$.
The main command can be typed by $(\tierc,\tierc,\tiera)$ as illustrated below, provided that $\cmd_1$ and $\cmd_2$ are the commands corresponding to the first while loop and second while loop, respectively.

\[
\begin{prooftree}
\hypo{\rho_1}
\ellipsis{}{}
\infer1[(W$_0$)]{\rho_1 \rightslice \vdash \cmd_1 : (\tierc,\tierc,\tiera)}
\hypo{\rho_2}
\ellipsis{}{}
\infer1[(W$_0$)]{\rho_2 \rightslice\vdash \cmd_2 : (\tierc,\tierc,\tiera)}
\infer2[(S)]{\vdash \cmd_1 \sap \cmd_2 : (\tierc,\tierc,\tiera)}
\end{prooftree}
\]

The typing derivation $\rho_1$ corresponds to the first while loop $\cmd_1$ and is described below.

\[
\scalebox{0.9}{
\begin{prooftree}
\hypo{\Gamma(\x)=\tierc}
\infer1[(V)]{\vdash \x : (\tierc,\tierc,\tierc)}
\infer1[(OP)]{\vdash \x>0 : (\tierc,\tierc,\tierc )}
\hypo{\rho_1^1}
\ellipsis{}{}
\infer1[(A)]{\vdash \x \asg \mpred(\x):(\tierc,\tierc,\tierc ) }
\hypo{\rho_2^1}sub
\ellipsis{}{}
\infer1[(A)]{\vdash \y \asg \msuc{1}(\msuc{1}(\y)) : (\tierb,\tierc,\tierc )}
\infer1[(SUB)]{\vdash \y \asg \msuc{1}(\msuc{1}(\y)) : (\tierc,\tierc,\tierc )}
\infer2[(S)]{\vdash\x \asg \mpred(\x) \sap \y \asg \msuc{1}(\msuc{1}(\y)) : (\tierc,\tierc,\tierc )}
\infer2[(W$_0$)]{\vdash \cmd_1 : (\tierc,\tierc,\tiera)}
\end{prooftree}
}
\]

The typing derivation $\rho_1^1$ can be built easily using rules (A), (OP), and (V) as $\mpred$ is neutral and can be 0.9given the type $\tierc \to \tierc$ in $\Delta(\mpred)(\tierc)$ (see Example~\ref{tre}). The typing derivation $\rho_2^1$ can be built using the same rules as $\msuc{1}$ is positive and can be given the type $\tierb \to \tierb$ in $\Delta(\msuc{1})(\tierc)$ (see Example~\ref{tre} again). $\rho_2^1$ requires the prior application of subtyping rule (SUB) as the tier of the assignment is equal to $\Gamma(\y)=\tierb$.

The typing derivation $\rho_2$, described below,

\[
\scalebox{0.9}{
\begin{prooftree}
\hypo{\Gamma(\y)=\tierb}
\infer1[(V)]{\vdash \y : (\tierb,\tierc,\tierb)}
\infer1[(OP)]{\vdash \y>0 : (\tierb,\tierc,\tierb )}
\hypo{\rho_1^2}
\ellipsis{}{}
\infer1[(A)]{\vdash \y \asg \mpred(\y):(\tierb,\tierb,\tierb ) }
\hypo{\rho_2^2}
\ellipsis{}{}
\infer1[(A)]{\vdash \z \asg \msuc{1}(\msuc{1}(\z)) : (\tiera,\tierb,\tierb )}
\infer1[(SUB)]{\vdash \z \asg \msuc{1}(\msuc{1}(\z)) : (\tierb,\tierb,\tierb )}
\infer2[(S)]{\vdash\y \asg \mpred(\y) \sap \z \asg \msuc{1}(\msuc{1}(\z)) : (\tierb,\tierb,\tierb )}
\infer2[(W$_0$)]{\vdash \cmd_2 : (\tierb,\tierc,\tiera)}
\infer1[(SUB)]{\vdash \cmd_2 : (\tierc,\tierc,\tiera)}
\end{prooftree}
}
\]

can be obtained in a similar way by taking the type $\tierb \to \tierb$ for the neutral operator $\mpred$ in $\Delta(\mpred)(\tierb)$ and the type $\tiera \to \tiera$ for the positive operator $\msuc{1}$ in $\Delta(\msuc{1})(\tierb)$. The initial subtyping rule is required as it is not possible to derive $\vdash \y>0 : (\tierc,\tierc,\tierc )$ with the requirement that $\Gamma(\y)=\tierb$.

It is worth noticing that the above program cannot be typed with only two tiers $\{\tiera,\tierb\}$. Indeed, the first while loop enforces that $\Gamma(\y) \ordst \Gamma(\x)$ and the second while loop enforces that $\Gamma(\z) \ordst \Gamma(\y)$. More generally, the program can be typed by $(\sla\mathbf{+2},\sla\mathbf{+2},\tiera)$ or $(\sla\mathbf{+2},\sla\mathbf{+2},\sla\mathbf{+2})$, for any tier $\sla$.
\end{exa}

\begin{exa}[Oracle]\label{oracles}
For a given input $\x$ and a given oracle $\phi$, the program below computes whether there exists a unary integer $n$ of size smaller than $\size{\x}$ such that $\phi(n) =0$.

\begin{lstlisting}
  $\y^\tiera \asg \x^\tierb \sap$
  $\z^\tiera \asg 0 \sap$
  $\while ( \x>= 0)^\tierb\{$
    $\ifa( \phi(\y^\tiera \upharpoonright \x^\tierb)\meq{}0)^\tiera\{$
      $\z^\tiera \asg 1$
    $\}$
   $\elsea\{\skp\} \sap$
    $\x^\tierb \asg \mpred(\x)^\tierb $
  $\}$
  $\ret\ \z$
\end{lstlisting}

This program is safe and can be typed by $(\tierb,\tierb,\tiera)$ under the variable typing environment $\Gamma$ such that $\Gamma(\x)=\tierb$ and $\Gamma(\y)=\Gamma(\z)=\tiera$.  The constants $0$ and $1$ can be considered to be neutral operators of zero arity and, hence, can be given any tier smaller than the innermost tier. It is easy to verify that the commands $\z \asg 1$, $\skp$, and $\x \asg \mpred(\x)$ can be typed by $(\tiera,\tierb,\tierb)$, $(\tiera,\tierb,\tierb)$, and $(\tierb,\tierb,\tierb)$, respectively, using typing rules (OP), (SK), and (A).

The conditional subcommand can be typed as described below.

\[
\scalebox{0.9}{
\begin{prooftree}
\hypo{\tiera \to \tierb \to \tiera \in \Delta(\meq{})(\tierb)}
\hypo{\Gamma(\y)=\tiera}
\infer1[(V)]{\vdash \y : (\tiera,\tierb, \tierb)}
\hypo{\Gamma(\x)=\tierb}
\infer1[(V)]{\vdash \x : (\tierb,\tierb, \tierb)}
\infer2[(OR)]{\vdash \phi(\y \upharpoonright \x) : (\tiera,\tierb,\tierb)}
\hypo{}
\infer1[(OP)]{\vdash 0 : (\tierb,\tierb,\tierb)}
\infer3[(OP)]{\vdash \phi(\y \upharpoonright \x) \meq{}0 : (\tiera,\tierb,\tierb)}
\hypo{\vdots}
\infer2[(C)]{\vdash \ifa( \phi(\y \upharpoonright \x)\meq{}0)\{ \z \asg 1\}\elsea\{\skp\} : (\tiera,\tierb,\tierb)}
\end{prooftree}
}
\]

The while loop will be typed using rule $(W_0)$. Consequently, the inner command can be typed by $(\tierb,\tierb,\tierb)$ after applying subtyping once.

Notice that the equivalent program obtained by swapping $\x$ and $\y$ in the oracle input (\ie, $\phi(\x \upharpoonright \y)$) is not typable as the tier of $\x$ has to be strictly smaller than the innermost tier in typing rule (OR). Although this requirement restricts the expressive power of the type system, it is strongly needed as it prevents uncontrolled loops on oracle outputs to occur. In particular, commands of the shape
$\while(\x>0)\{ \x \asg \phi(\x \upharpoonright \x)\}$
are rejected by the type system.

Note that the above program is typable as the oracle calls are performed in a decreasing order and, hence, does not break the finite lookahead revision property, which will be presented in \S\ref{ss:flr}.

Now consider the equivalent program where oracle calls are performed in increasing order.

\begin{lstlisting}
  $\x \asg 0 \sap$
  $\z \asg 0 \sap$
  $\while ( \y>= \x)\{$
    $\ifa( \phi(\y \upharpoonright \x)\meq{}0)\{$
      $\z \asg 1$
    $\}$
   $\elsea\{\skp\} \sap$
    $\x \asg \msuc{1}(\x) $
 $\}$
 $\ret\ \z$
\end{lstlisting}

This program is not a safe program. Suppose, by contradiction, that it can be typed with respect to a safe operator typing environment. The innermost tier $\sla$ of the commands under the while will be equal to the tier of the guard  $\y>= \x$, independently of whether rule (W) or rule (W$_0$) is used to type the while command. Moreover, $\x$ has a tier $\sla_\x$ such that $\sla \ord \sla_\x$, using rule (OP) on the guard and, by definition of safe typing environments.
Now $\msuc{1}$ is a positive operator and, consequently, by rule (OP) and, by definition of safe typing environments again, $\msuc{1}(\x)$ has a tier $\sla_{\msuc{1}(\x)}$ strictly smaller than the innermost tier, \ie, $\sla_{\msuc{1}(\x)} \ordst \sla$. By typing rule (A), in order to be typed, the command $\x \asg \msuc{1}(\x)$ enforces  $\sla_\x \ord \sla_{\msuc{1}(\x)}$. Hence, we obtain a contradiction: $\sla \ordst \sla$.
\end{exa}

\begin{exa}[Multiple tiers with oracle]
The following program
computes the function $\Sigma_{i=0}^{\max_{x=0}^{n}\phi(x)}\phi(i)$.

\begin{lstlisting}
  $\x^\tierd \asg n \sap \y^\tierc \asg \x^\tierd \sap \z^\tierc  \asg 0 \sap$
  $\while ( \x^\tierd>= 0)\{$
    $\z^\tierc \asg \max(\phi(\y^\tierc \upharpoonright \x^\tierd)^\tierc,\z^\tierc)\sap$
    $\x^\tierd \asg \mpred(\x^\tierd)$
  $\}\sap \vvv^\tierb \asg \z^\tierc \sap \uu^\tiera \asg 0\sap$
  $\while ( \z^\tierc>= 0)\{$
    $\ww^\tierb \asg \phi(\vvv^\tierb \upharpoonright \z^\tierc)^\tierb \sap$
    $\while ( \ww^\tierb>= 0)\{$
      $\uu^\tiera \asg \msuc{1}(\uu)^\tiera\sap \ww^\tierb \asg \mpred(\ww^\tierb)$
    $\}\sap \z ^\tierc \asg \mpred(\z)^\tierc$
  $\}$
  $\ret\ \uu$
\end{lstlisting}

This program can be typed by $(\tierd,\tierd,\tiera)$ under the variable type assignment $\Gamma$ such that $\Gamma(\x)=\tierd$, $\Gamma(\y)=\Gamma(\z)=\tierc$, $\Gamma(\vvv)=\Gamma(\ww)=\tierb$, and $\Gamma(\uu)=\tiera$.

The first while loop will be typed using rule (W$_0$). Consequently, its inner command is typed by $(\tierd,\tierd,\tierd)$. As the $\max$ operator is neutral, it can be given the type $\tierc \to \tierc \to \tierc \in \Delta(\max)(\tierd)$. The oracle call is typable as the input data $\y$ has a tier strictly small than the innermost tier ($\tierd$) and the input bound has tier equal to the outermost tier ($\tierd$).

The second while loop can be typed using rule (W$_0$) after applying subtyping rule (SUB). Consequently, its  inner command is typed by $(\tierc,\tierc,\tierc)$. The oracle call is performed on input data of strictly smaller tier ($\tierb$) and on input bound of tier equal to the outermost tier ($\tierc$). The inner while loop can be typed using rule (W) and thus updates the innermost tier to $\tierb$. Consequently, $\msuc{1}$ is enforced to be of tier $\tiera \to \tiera$ in the inner command.
\end{exa}

\section{Properties of safe programs}\label{s:prop}
We now show the main properties of safe programs:
\begin{itemize}
\item a standard non-interference property in \S\ref{ss:ni} ensuring that computations on higher tiers do not depend on lower tiers (Theorem~\ref{thm:ni}).
\item a polynomial time property in \S\ref{ss:psc} ensuring that terminating programs terminate in time polynomial in the input size and maximal oracle output size (Theorem~\ref{thm:pol}).
\item a finite lookahead revision property in \S\ref{ss:flr} ensuring that, for any oracle and any input, the number of oracle calls on input of increasing size is bounded by a constant (Theorem~\ref{thm:flr}).
\end{itemize}

\subsection{Notation}
Let us first introduce some preliminary notation. Let $\E(a)$ (res. $\Cb(a)$) be the set of expressions (respectively commands) occurring in $a$, for $a \in \{\prog,\cmd\}$. Let $\A(\cmd)$ be the set of variables that are assigned to in $\cmd$, \eg, $\A(\x \asg \y \sap \y \asg \z)= \{\x,\y\}$. Let $\Oo(\e)$ and $\FV(\e)$ be the set of operators in expression $\e$ and the set of variables in expression $\e$, respectively.

\subsection{Non-interference}\label{ss:ni}
We now show that the type system provides classical non-interference properties.

In a safe program, only variables of tier higher than $\sla$ can be accessed to evaluate an expression of tier $\sla$.
\begin{lem}[Simple security]\label{lem:ss}
Given a safe program $\prog $ with respect to the typing environments $\typenv,\typop$, for any expression $\e \in \E(\prog)$, if $\ \pbl \e : (\sla,\sla_{in},\sla_{out})$, then  for all $ \x \in \FV(\e)$, $ \sla \ord \Gamma(\x)$.

\end{lem}
\begin{proof}
 By structural induction on expressions.\qedhere
 \end{proof}
There is no equivalent lemma for commands because of the subtyping rule (SUB).

\begin{cor}\label{coro:ss}
Given a safe program $\prog $ with respect to the typing environments $\typenv,\typop$,
for any $\x \asg \e \in \Cb(\prog)$, $\Gamma(\x) \ord  \meet_{\y \in \FV(\e)}\Gamma(\y)$.
\end{cor}
\begin{proof}Given a command $\x \asg \e \in  \Cb(\prog)$ of a safe program with respect to the typing environments $\typenv,\typop$, there are $\sla_1, \sla_2,\sla_{in},\sla_{out}$ such that $\pbl \x: (\sla_1,\sla_{in},\sla_{out})$, $\pbl \e: (\sla_2,\sla_{in},\sla_{out})$ and $\sla_1  \ord \sla_2$, by typing rule (A). Applying Lemma~\ref{lem:ss}, $\forall \y \in \FV(\e)$, $\sla_2 \ord \Gamma(\y)$. By typing rule (V) $\sla_1 = \Gamma(\x)$ and, consequently, $\forall \y \in \FV(\e)$, $\Gamma(\x) \ord \Gamma(\y)$.
 \qedhere
 \end{proof}

The following Lemma states that command tiers are monotonic in their subcommand tier in a given typing derivation.
\begin{lem}\label{lem:strat}
Let $\rho$ be a typing derivation of a safe program  with respect to the typing environments $\typenv,\typop$. For any typing derivations $\rho_1 \rightslice \pbl \cmd_1 : (\sla^1,\sla^1_{in},\sla^1_{out}) \in \D(\rho)$ and $\rho_2 \rightslice \pbl \cmd_2 : (\sla^2,\sla^2_{in},\sla^2_{out})\in \D(\rho_1)$, $\sla^2 \ord \sla^1$.
\end{lem}
\begin{proof}
Suppose by contradiction that $\rho_1 <\rho_2$ and $\sla^1 \ordst \sla^2$ hold. As all typing rules for commands in Figure~\ref{TS} are monotonic in the command tier,  $\rho_2 \rightslice \pbl \cmd_2 : (\sla^2,\sla^2_{in},\sla^2_{out})$ cannot be derived from $\rho_1$.\qedhere
\end{proof}

The confinement Lemma expresses the fact that commands of tier $\sla$ cannot write in variables of strictly higher tier.

\begin{lem}[Confinement]\label{lem:confinement}
Given a safe program $\prog $ with respect to the typing environments $\typenv,\typop$,
for any $\cmd \in  \Cb(\prog)$, if $\ \pbl \cmd:(\sla,\sla_{in},\sla_{out})$, then  for all $ \x \in \A(\cmd)$, $\Gamma(\x) \ord \sla$.
\end{lem}

\begin{proof}
By contradiction. Suppose that $\rho_1 \rightslice \pbl \cmd:(\sla,\sla_{in},\sla_{out})$ holds. Consider a variable $\x \in \A(\cmd)$ and suppose that $\sla \ordst \Gamma(\x)$. By typing rule (A), there is an expression $\e$  and there are tiers $\sla_{in}',\sla_{out}'$ such that $\rho_2 \rightslice  \pbl \x \asg \e : (\Gamma(\x), \sla_{in}',\sla_{out}')$ and $\rho_2 \in \D(\rho_1)$.
Consequently, by Lemma~\ref{lem:strat}, $ \Gamma(\x) \ord \sla$, which contradicts the assumption.
\qedhere
\end{proof}

For a given variable typing environment $\Gamma$ and a given tier $\sla$, we define an equivalence relation on stores by:
$\store \approx^{\Gamma}_{ \sla }\store'$ if $\dom(\store) = \dom(\store')=\dom(\Gamma)$ and  for each $ \x \in \dom(\Gamma)$, if $\sla \ord \Gamma(\x)$ then $\store(\x)=\store'(\x)$.

We introduce a non-interference Theorem ensuring that the values of tier $\sla$ variables during the evaluation of a program do not depend on values of tier $\sla'$ variables for $\sla' \ordst \sla$.

\begin{thm}[Non-interference]\label{thm:ni}
Given a safe program $ \cmd\ \ret \x $ with respect to the typing environments $\typenv,\typop$. For any stores $\store_1$ and $\store_2$ if $\store_1 \approx^{\Gamma}_{ \sla }\store_2$,
 $ \store_1\Imp \cmd  \to  \store'_1$ and $ \store_2\Imp \cmd  \to  \store'_2$
  then $\store'_1 \approx^{\Gamma}_{ \sla }\store'_2$.
\end{thm}

\begin{proof}
By structural induction on derivations. The base case of a derivation consisting in only one node is straightforward as only the rule (Skip) of Figure~\ref{fig:Com} can be fired. Consequently, $\cmd = \skp$ and $\store'_1=\store_1 \approx^{\Gamma}_{ \sla } \store'_2=\store_2$.

Now consider the two derivations $ \pi_\phi^1 : \store_1\Imp \cmd  \to  \store'_1$ and $\pi_\phi^2 : \store_2\Imp \cmd  \to  \store'_2$ such that $\store_1 \approx^{\Gamma}_{ \sla }\store_2$. We perform a case analysis on commands.
\begin{itemize}
\item if $\cmd = \x \asg \e$ then the rule at the root of derivations $ \pi_\phi^1 $ and $ \pi_\phi^2$ is the rule (Asg) of Figure~\ref{fig:Com}. There are two cases to consider.
\begin{itemize}
\item If $\sla \ord \Gamma(\x)$ then, by Corollary~\ref{coro:ss}, $\forall \y \in \FV(\e)$, $\Gamma(\x) \ord \Gamma(\y)$ and consequently, $\forall \y \in \FV(\e)$, $\store_1(\y)=\store_2(\y)$. This implies that there exists $\w \in \W$ such that $\store_1 \Imp \e \to w$ and $\store_2 \Imp \e \to w$. Consequently $\store_1 \Imp \cmd \to \store_1[\x \leftarrow \w]$, $\store_2 \Imp \cmd \to \store_2[\x \leftarrow \w]$ and  $\store_1[\x \leftarrow \w] \approx^{\Gamma}_{ \sla }  \store_2[\x \leftarrow \w]$.
\item  If $\Gamma(\x) \ordst \sla$ then $\store_1 \Imp \cmd \to \store_1[\x \leftarrow \w]$, $\store_2 \Imp \cmd \to \store_2[\x \leftarrow \varv]$ and $\store_1[\x \leftarrow \w] \approx^{\Gamma}_{ \sla }  \store_2[\x \leftarrow \varv]$.
\end{itemize}
\item if $\cmd = \cmd_1 \sap \cmd_2$ then the rule at the root of derivations $ \pi_\phi^1 $ and $ \pi_\phi^2$ is the rule (Seq) of Figure~\ref{fig:Com}. Consequently, there exist two stores $\store''_1$ and $\store''_2$, such that, for $i \in \{1,2\}$, $\pi_\phi^{i'}: \store_i \Imp \cmd_1 \to \store''_i$ and $\pi_\phi^{i''}:\store''_i \Imp \cmd_2 \to \store'_i$ are subderivations of $\pi_\phi^i $. Therefore, if $\store_1 \approx^{\Gamma}_{ \sla }  \store_2$ then $\store''_1 \approx^{\Gamma}_{ \sla } \store''_2$ and $\store'_1 \approx^{\Gamma}_{ \sla } \store'_2$, by applying the induction hypothesis twice.
\item if $\cmd = \while (\e) \{\cmd'\}$ then the rule at the root of derivations $ \pi_\phi^1 $ and $ \pi_\phi^2$ can be either rule (Wh$_0$) or rule (Wh$_1$) of Figure~\ref{fig:Com}. If $ \pbl  \while (\e) \{\cmd'\} :(\sla_1,\sla_{in},\sla_{out})$ then $ \pbl  \e :(\sla_1,\sla_{in},\sla_{out})$ and there are two cases to consider:
\begin{itemize}
\item If $\sla \ord \sla_1$ then, by Lemma~\ref{lem:ss}, $ \forall \x \in \FV(\e),\ \sla_1 \ord \Gamma(\x)$. Consequently, $ \forall \x \in \FV(\e), \ \store_1(\x)=\store_2(\x)$. It follows that $\store_1 \Imp \e \to w$ and $\store_2 \Imp \e \to w$, with $w \in \{0,1\}$. Notice that,  it excludes the possibility to apply (Wh$_1$) on one derivation and (Wh$_0$) on the other derivation. For the non-trivial case where $w =1$, we obtain that $\pi_\phi^{1'} :\store_1 \Imp \cmd' \sap \cmd \to \store_1'$, $\pi_\phi^{2'} : \store_2 \Imp \cmd' \sap \cmd \to \store_2'$, and $\store'_1 \approx^{\Gamma}_{ \sla }  \store'_2$, by induction on the subderivations $\pi_\phi^{1'}$ and $\pi_\phi^{2'}$. Consequently, $\store_1 \Imp \cmd \to \store_1'$, $\store_2 \Imp  \cmd \to \store_2'$ and $\store'_1 \approx^{\Gamma}_{ \sla }  \store'_2$.
\item If $ \sla_1 \ordst \sla$ then, by Lemma~\ref{lem:confinement}, $\forall \x \in \A(\cmd'),\  \Gamma(\x) \ord \sla_1$. Consequently, if $\store_1 \Imp \cmd \to \store'_1$ and $\store_2 \Imp \cmd \to \store'_2$ then $\store'_1 \approx^{\Gamma}_{ \sla } \store'_2$.
\end{itemize}
All the other cases can be treated in a similar manner.\qedhere
\end{itemize}
\end{proof}

\subsection{Polynomial step count}\label{ss:psc}
In this section, we show that terminating and safe programs have a runtime polynomially bounded by the size of the input store and the maximal size of answers returned by the oracle in the course of execution.

The following Lemma shows that the innermost tier of a while loop subcommand is always an upper bound on the tier of this loop.
\begin{lem}\label{lem:inner}
Let $\rho$ be a typing derivation of a safe program  with respect to the typing environments $\typenv,\typop$. For any typing derivations  $ \rho_1 \rightslice \pbl \while (\e_1) \{\cmd_1\}   : (\sla^1,\sla^1_{in},\sla^1_{out})\in \D(\rho)$ and $\rho_2 \rightslice \pbl \cmd_2 : (\sla^2,\sla^2_{in},\sla^2_{out}) \in \mathring{\D}(\rho_1)$, $\sla^2_{in} \ord \sla^1$.
\end{lem}
\begin{proof}
Given a typing derivation $\rho_1 \rightslice \pbl \while (\e_1) \{\cmd_1\} : (\sla^1,\sla^1_{in},\sla^1_{out})$, we show by induction that the property $Inv$ defined on command typing derivations by:

 \[Inv(\rho \rightslice \pbl \cmd : (\sla,\sla_{in},\sla_{out}))\ \triangleq \ (\sla_{in} \ord \sla^1)\]
 is invariant on command typing derivations strictly smaller than $\rho_1$, \ie, typing derivations $\rho$ such that $\rho < \rho_1$.

\emph{Base case.} Suppose that the typing derivation $\rho_1$ is of the shape:

\[
\begin{prooftree}
\hypo{\pbl \e_1:  \ldots}
\hypo{\rho_2 \rightslice \pbl \cmd_1 : (\sla^2,\sla^2_{in},\sla^2_{out})}
\infer2[(R)]{\rho_1 \rightslice \pbl  \while (\e_1) \{\cmd_1\} :(\sla^1,\sla^1_{in},\sla^1_{out})}
\end{prooftree}.
\]

Rule (R) can be either rule (W) or rule (W$_0$) of Figure~\ref{TS}. In both cases, it holds that $\sla^2=\sla^2_{in}=\sla^1$ and, consequently, $Inv(\rho_2)$ holds.

\emph{General case.} Consider the typing derivation $\rho_3 \in \mathring{\D}(\rho_2)$ of the shape:

\[
\begin{prooftree}
\hypo{\ldots}
\hypo{\rho_4 \rightslice \pbl \cmd_4 : (\sla^4,\sla^4_{in},\sla^4_{out})}
\infer2[(R')]{\rho_3 \rightslice \pbl  \cmd_3 :(\sla^3,\sla^3_{in},\sla^3_{out})}
\end{prooftree},
\]
for some typing rule (R'). By induction hypothesis, $Inv(\rho_3)$ holds, and, consequently, $\sla^3_{in} \ord \sla^1$. Moreover, by Lemma~\ref{lem:strat}, $\sla^3 \ord \sla^1$ holds.

 By a case analysis on typing rules of Figure~\ref{TS},
 (R') can be either (S), (A), (SUB), (C), or (W). Indeed, rule (W$_0$) cannot be applied under a while loop as it requires the outermost tier to be equal to $\tiera$. This is impossible under a while loop as, by looking at the constraints in rules (W) and (W$_0$) of Figure~\ref{TS}, one can check that the outermost tier of a while loop strict subcommand is greater than $\tierb$.

These five rules imply that either $\sla^4_{in} = \sla^3_{in}$ (rules (S), (A), (SUB), and (C)) or $\sla^4_{in} = \sla^3$ (rule (W)). Consequently, $\sla^4_{in} \ord \sla^1$, and $Inv(\rho_4)$ holds.
\qedhere
\end{proof}

Consequently, the innermost tier of a while loop provides a lower bound on the tier of the loop guard expression.

We now show that within a while loop of tier $\sla$, if an expression is assigned to a tier $\sla$ variable then this expression does not contain positive operators. Moreover, there are no oracle calls in tier $\sla$ subcommands of a tier $\sla$ while loop.

\begin{lem}\label{lem:stratification}
Let $\rho$ be a typing derivation of a safe program with respect to the typing environments $\typenv,\typop$.
For any $\rho_1 \rightslice \pbl  \while (\e_1) \{\cmd_1\} : (\sla^1,\sla^1_{in},\sla^1_{out})\in \D(\rho)$ and $ \rho_2 \rightslice \pbl \x \asg \e_2 : (\sla^2,\sla^2_{in},\sla^2_{out}) \in \D(\rho_1)$, if $\sla^1 =\sla^2$ then:
\begin{enumerate}
\item  for any $ \op \in \Oo(\e_2)$, $\op$ is a neutral operator;
\item  there is no oracle call in $\e_2$.
\end{enumerate}
\end{lem}
\begin{proof}
Both (1) and (2) are proved by contradiction using Lemma~\ref{lem:inner}.

\begin{figure*}
\hrulefill
\\[10pt]
\centering
\begin{prooftree}
\hypo{\Gamma(\x) = \sla^2}
\infer1[(V)]{\pbl \x : (\sla^2,\sla^2_{in},\sla^2_{out})}
\hypo{\sla^4  \to  \sla^3 \in \Delta(\op)(\sla^2_{in})}
\hypo{\Gamma(\y)= \sla^4}
\infer1[(V)]{\pbl \y : (\sla^4,\sla^2_{in},\sla^2_{out})}
\infer2[(OP)]{\op(\y) : (\sla^3 ,\sla^2_{in},\sla^2_{out}) }
\infer2[(A)]{\pbl \x \asg \op(\y) : (\sla^2 ,\sla^2_{in},\sla^2_{out})}
\end{prooftree}
\\[10pt]
\hrulefill
\caption{Proof of Lemma~\ref{lem:stratification} (1)}\label{fig:stratification(1)}
\end{figure*}

(1) By contradiction, suppose that there exists an assignment $\x \asg \e_2  $ such that $\sla^1=\sla^2$ and there exists an operator in $\Oo(\e_2)$ that is positive. For simplicity, suppose that $ar(\op)=1$ and that $\e_2 = \op(\y)$. We obtain the typing derivation of Figure~\ref{fig:stratification(1)},
with the constraints that $\sla^2 \ord \sla^3 \ord \sla^4 \ord \sla^2_{in}$. $\sla^2 \ord \sla^3$  is enforced by rule (A) and $\sla^3 \ord \sla^4 \ord \sla^2_{in}$ is enforced by definition of safe typing environment applied to $\Delta(\op)(\sla^2_{in})$ in rule (OP).  Moreover as $\op$ is positive, the constraint $\sla^3  \ordst \sla^2_{in}$ is enforced and, consequently, $\sla^2 \ordst \sla^2_{in}$. However, by Lemma~\ref{lem:inner} applied to the typing derivation sub-tree of root $\x \asg \op(\y)$, $\sla^2_{in} \ord \sla^1$, hence $\sla^2_{in} \ord \sla^2$, by hypothesis ($\sla^2=\sla^1$), and we obtain a contradiction.
The general case can be treated by a structural induction on expressions.

(2) By contradiction suppose that $\sla^1 =\sla^2$ and that there is an oracle call in $\e_2$.
Suppose for simplicity that $\e_2=\phi(\e' \upharpoonright \e'')$.  We obtain the typing derivation of Figure~\ref{fig:stratification(2)}.
\begin{figure*}
\hrulefill
\\[10pt]
\centering
\begin{prooftree}
\hypo{\Gamma(\x) = \sla^2}
\infer1[(V)]{\pbl \x : (\sla^2,\sla^2_{in},\sla^2_{out})}
\hypo{}
\ellipsis{}{}
\infer1[]{\pbl \e' : (\sla^3 ,\sla^2_{in},\sla^2_{out}) }
\hypo{}
\ellipsis{}{}
\infer1[]{\pbl \e'' : (\sla^2_{out} ,\sla^2_{in},\sla^2_{out})}
\infer2[]{\pbl \phi(\e' \upharpoonright \e'') : (\sla^3 ,\sla^2_{in},\sla^2_{out}) }
\infer2[(A)]{\pbl \x \asg \phi(\e' \upharpoonright \e'') : (\sla^2 ,\sla^2_{in},\sla^2_{out})}
\end{prooftree}
\\[10pt]
\hrulefill
\caption{Proof of Lemma~\ref{lem:stratification} (2)}\label{fig:stratification(2)}
\end{figure*}
By the constraint of typing rule (A), $\sla^2 \ord \sla^3$. By the constraints of typing rule (OR), $\sla^3 \ordst \sla^2_{in}$ and $\sla^3 \ord \sla^2_{out}$.
Consequently,  $\sla^2 \ordst \sla^2_{in}$ is enforced. By Lemma~\ref{lem:inner}, $ \sla^2_{in} \ord \sla^1$ and, consequently, we obtain a contradiction.
The general case can be treated in a similar manner by a structural induction on expressions.
\qedhere
\end{proof}

\begin{defi}
Let $m_{\store}^{\prog}$ be the maximum of $\size{\store}$ and the maximum size of an oracle answer in the derivation $\pi_{\phi} : \store \Imp \prog \to \varu$.
Formally, \[m_{\store}^{\prog} =\max_{(\varv,\w) \in C(\pi_\phi)}(\size{\store}, \max\{\size{\phi(\sem{\upharpoonright}(\varv,\w))} \}),\]
where $C(\pi_\phi)$ is the set of pairs $(\varv,\w)$ such that \[
\begin{prooftree}\hypo{\store \Imp \e_1 \to \varv \quad\store \Imp \e_2 \to \w }\infer1[]{\store \Imp \phi(\e_1\upharpoonright \e_2) \to \phi(\sem{\upharpoonright}(\varv,\w))} \end{prooftree} \in \pi_{\phi}.\]

A program $\prog$ has a \emph{polynomial step count} if there is a polynomial $P$ such that for any store $\store$ and any oracle $\phi$, $\pi_{\phi} : \store \Imp \prog \to \w$, $\size{\pi_{\phi}} \leq  P (m_{\store}^{\prog})$.
\end{defi}

We show that a safe program has a polynomial step count on terminating computations.

\begin{thm}\label{thm:pol}
Given a safe program $\prog$ with respect to the typing environments $\typenv,\typop$, there is a polynomial $P$ such that for any derivation $\pi_{\phi} : \store \Imp \prog \to \w$, $\size{\pi_{\phi}} \leq P (m_{\store}^{\prog})$.
\end{thm}

\begin{proof}
By induction on the tier of a command using non-interference Theorem (\ref{thm:ni}) and Lemma~\ref{lem:stratification}.

The case $\tiera$ is trivial as there are no while loops (see rules (W) and (W$_0$) of Figure~\ref{TS}). Now consider a safe program, whose command is of tier $\sla$ strictly greater than $\tiera$. For any variable of tier $\sla$, by Lemma~\ref{lem:stratification}, only neutral operators may be applied in assignments to such a variable in a while. Combining Definition~\ref{def:np} and non-interference (Theorem~\ref{thm:ni}), the word computed by a neutral operator is either a subword of the initial values stored in tier $\sla$ variables or a constant in $\{0,1\}$. Moreover, the number of times a positive operator is applied to such a variable is at most constant. Indeed, such an assignment occurs outside a while loop. Putting it altogether, the number of distinct values in a store for each tier $\sla$ variable is in $O((m_{\store}^{\prog})^2)$ as strict subword operations can be iterated at most a quadratic number of time.

 The evaluation of while loops of tier $\sla$ unfolds at most a polynomial number of commands of tier at most $\sla-\tierb$. The degree of the polynomial depends on the number of variables of tier $\geq \sla$, by non-interference (Theorem~\ref{thm:ni}), and on the number of nested tier $\sla$ while loops and, as a consequence, is bounded by a constant: the size of the program. By induction, all these commands of tier strictly smaller are evaluated in a polynomial number of steps in their input. Each of their input is polynomially bounded by $m_{\store}^{\prog}$ as it consists in oracle calls combined with at most a polynomial number of positive operators.

It just remains to observe that we compose a constant number ($\sla\mathbf{+1}$) of polynomials and that polynomials are closed under composition.
\end{proof}

The proof of the above Theorem is similar to proofs of polynomiality in~\cite{M11} and~\cite{MP14}, except for two distinctions:
\begin{itemize}
\item As strictly more than 2 tiers are allowed, the innermost tier $\sla_{in}$ is used to ensure that operators and oracle calls are stratified (Lemma~\ref{lem:stratification}): in a while loop of innermost tier $\sla_{in}$ the return type of an oracle or positive operator is always strictly smaller than $\sla_{in}$. Hence the results of such computations cannot be assigned to variables whose tier is equal to $\sla_{in}$.
\item Oracle calls may return a value whose size is not bounded by the program input. This is the reason why $m_{\store}^{\prog}$ has to be considered as an input of the time bound.
 \end{itemize}
\begin{cor}\label{psc}
Given a program $\prog$,  if $\prog \in \st$ then $\prog$ has a polynomial step count.
\end{cor}

\subsection{Finite lookahead revision}\label{ss:flr}
In this section, we show that, whereas terminating and safe programs may perform a polynomial number (in the size of the input and the maximal size of the oracle answers) of oracle calls during their execution, they may only perform a constant number of oracle calls on input data of increasing size.

We first start to show that the outermost tier of a command of a safe program is an upper bound on the tiers of while loop expressions guarding this command.
\begin{lem}\label{lem:kout}
Let $\rho$ be a typing derivation of a safe program with respect to the typing environments $\typenv,\typop$.
For any $\rho_1 \rightslice \pbl  \while (\e_1) \{\cmd_1\} : (\sla^1,\sla^1_{in},\sla^1_{out}) (R) \in \D(\rho)$ and $ \rho_2 \rightslice \pbl \cmd_2 : (\sla^2,\sla^2_{in},\sla^2_{out}) \in \mathring{\D}(\rho_1)$,
if $R \in \{W,W_0\}$ then $\sla^1 \ord \sla^2_{out} $.
\footnote{The assumption for the last rule of $\rho_1$ to be in $\{$W,W$_0\}$ ensures that $\sla^1$ is exactly the tier of the guard expression $\e_1$. Otherwise typing rule (SUB) can be applied arbitrarily and the result does not hold.}
\end{lem}

\begin{proof}
The typing rule (W$_0$) is the only rule changing the outermost tier. It is straightforward to observe that this rule can only be applied to an outermost while loop as the outermost tier of the command is updated from $\tiera$ to $\sla$, for some $\sla$ such that $\tierb \ord \sla$.
There are two cases to consider depending on $R$:
\begin{itemize}
\item if R=W then

\[
\begin{prooftree}
\hypo{\ldots}
\hypo{\rho'_1 \rightslice \pbl \cmd_1 : (\sla^1,\sla^1,\sla^1_{out})}
\hypo{\tierb \ord \sla^1 \ord \sla^1_{out}}
\infer3[(W)]{\rho_1 \rightslice \pbl  \while (\e_1) \{\cmd_1\}\ :(\sla^1,\sla^1_{in},\sla^1_{out})}
\end{prooftree}.
\]

Clearly, $\sla^1 \ord \sla^1_{out}= \sla^2_{out}$ by the guard condition and as $\rho_2 \leq \rho'_1$ and all the rules under a while preserve the outermost tier.
\item if R=W$_0$ then

\[
\begin{prooftree}
\hypo{\ldots}
\hypo{\rho'_1 \rightslice \pbl \cmd_1 : (\sla^1,\sla^1,\sla^1 )}
\hypo{\tierb \ord \sla^1  }
\infer3[(W$_0$)]{\rho_1 \rightslice \pbl  \while (\e_1) \{\cmd_1\}\ :(\sla^1,\sla^1_{in},\tiera)}
\end{prooftree}.
\]

It is straightforward that $\sla^1 = \sla^2_{out}$  as $\rho_2 \leq \rho'_1$ ($\rho_2 \in \mathring{\D}(\rho_1)$) and all the rules under a while preserve the outermost tier.\qedhere
\end{itemize}
\end{proof}

\begin{defi}
Given a program $\prog$, let $(l^{\pi_{\phi}}_n)$ be the sequence of oracle input values $\sem{\upharpoonright}(\varv,\w)$ in a rule (OR) obtained by a left-to-right depth-first traversal of the derivation $\pi_{\phi} : \store \Imp \prog \to \w$. Let $lr((l^{\pi_{\phi}}_n)) =\#\{i \ | \ \size{l_{i}^{\pi_{\phi}}} >\max_{j < i}(\size{l_j^{\pi_{\phi}}})\}$.

$\prog$ has  \emph{finite lookahead revision} if there is a constant $r$ such that for any oracle  $\phi$ and for any derivation $\pi_{\phi}$ (\ie, for all program inputs), we have $lr((l^{\pi_{\phi}}_n)) \leq r$.
\end{defi}

Note that the left-to-right depth-first traversal in a derivation exactly corresponds to the order of a sequential execution of a command.

\begin{thm}[Finite lookahead revision]\label{thm:flr}
Given a program $\prog$, if $\prog$ is safe with respect to the typing environments $\typenv,\typop$ then it has  finite lookahead revision.
\end{thm}

\begin{proof}
By Lemma~\ref{lem:kout}, the outermost tier $\sla_{out}$ of any command within a while loop command $\cmd$ is an upper bound on the tier of any while loop guard in $\cmd$. By typing rule (OR) of Figure~\ref{TS}, the tier of the input bound $\e_2$ in any oracle call $\phi(\e_1\upharpoonright \e_2)$ in $\cmd$ must be equal to the outermost tier $\sla_{out}$. By Lemma~\ref{lem:ss}, only variables of tier greater than or equal to $\sla_{out}$ can occur in $\e_2$. Consider such a variable $\x$ of tier $\sla$. For simplicity, we assume that $\x$ is the only variable of tier greater than or equal to $\sla_{out}$ in $\cmd$.
If $\sla_{out} \ordst \sla$ then $\x$ cannot be assigned to in $\cmd$, by Lemma~\ref{lem:confinement}. If $\sla = \sla_{out}$ then $\x$ can only be assigned to by expressions that contain neither oracle calls nor positive operators, by Lemma~\ref{lem:stratification}.
It means that in any assignment $\x = \e$ of the loop, the expression $\e$ can only contain variable $\x$ and neutral operators. Hence the length of the value stored by $\x$ cannot increase between two iterations of any loop in $\cmd$. 
Following the same reasoning, the length of the oracle's input bound $\e_2$ cannot increase between two iterations of a loop in $\cmd$. As $\size{\sem{\upharpoonright}(v,w)}=\size{\w}$+1, the length of the oracle's input values cannot increase within a loop in $\cmd$ and, consequently, the number of lookahead revisions in a sequence $(l_n^{\pi_{\phi}})$ is bounded by a constant.\footnote{Intuitively, this constant corresponds to the number of consecutive non-nested while loops.}

The general case, where several variables of tier greater than or equal to $\sla_{out}$ occur inside a while loop, only differs by a constant factor as only a finite number of different assignments on these variables may happen. \qedhere
\end{proof}

\section{Soundness}\label{s:sound}
In this section, we show a soundness result: the type-2 simply typed lambda-closure of programs in $\st$ is included in the class of basic feasible functionals $\BFF_2$~\cite{M76,KC91,KC96}. For that purpose, we use the characterization of~\cite{KS18} based on moderately polynomial time functionals. We show that terminating and safe program can be simulated by oracle Turing machines with a polynomial step count and a finite lookahead revision. We discuss briefly the requirement of the lambda-closure in \S\ref{ss:pos}.

\subsection{Moderately Polynomial Time Functionals}
We consider oracle Turing machines $M_{\phi}$ with one query tape and one answer tape for oracle calls. 
If a query is written on the query tape and the machine enters a \emph{query state}, then the oracle's answer appears on the answer tape in one step.

\begin{defi}
Given an oracle TM $M_\phi$, computing a total function, and an input $\textbf{a}$, let $m_{\textbf{a}}^{M_\phi}$ be the maximum of the size of the input $\textbf{a}$ and of the biggest oracle answer in the run of machine on input $\textbf{a}$ with oracle $\phi$.
A machine $M_{\phi}$ has:
\begin{itemize}
\item a polynomial step count if there is a polynomial  $P$ such that for any input $\textbf{a}$ and oracle $\phi$, $M$ runs in time bounded by $P(m_{\textbf{a}}^{M_ \phi})$.
\item finite lookahead revision if there exists a natural number $r \in \mathbb{N}$ such that for any oracle and any input, in the run of the machine, it happens at most $r$ times that a query is posed whose size exceeds the size of all previous queries.
\end{itemize}
\end{defi}

\begin{defi}[Moderately Polynomial Time]
$\mpt$ is the class of second order functionals computable by an oracle TM with a polynomial step count and finite lookahead revision.
\end{defi}
The set of functionals over words $\W$ is defined to be the set of functions of type $\tau_1 \to \ldots \to \tau_n \to \W$, where the each type $\tau_i$ is defined inductively by $\tau ::= \W \ | \ \tau \to \tau $.

Suppose given a countably infinite number of variables $x^\tau,y^\tau,\ldots$, for each type $\tau$.
For a given class of functionals $X$, let $\lambda(X)$ be the set of closed simply typed lambda-terms generated inductively as follows:
\begin{itemize}
\item for each type $\tau$, variables $x^\tau, y^\tau, \ldots$ are terms,
\item each functional $F \in X$ of type $\tau$ is a term,
\item for any term $t$ of type $\tau'$ and variable $x^{\tau}$, $\lambda x.t$ is a term of type $\tau \to \tau'$,
\item for any terms $t$ of type $\tau \to \tau'$ and $s$ of type $\tau$, $t\ s$ is a term of type $\tau'$.
\end{itemize}
Each lambda-term of type $\tau$ represents a functional of type $\tau$ and terms are considered up to $\beta$ and $\eta$ equivalence.
The level of a type is defined inductively by $lev(\W)=0$ and $lev(\tau \to \tau')=\max(lev(\tau)+1,lev(\tau'))$. For a given class of functionals $X$, let $X_2$ be the set of functionals of level $2$.

\begin{lem}[Monotonicity]\label{closure}
Given two classes $X,Y$ of functionals,if $X \subseteq Y$ then $\lambda(X)_2 \subseteq \lambda(Y)_2$.
\end{lem}

For a given functional $F$ of type $\tau_1 \to \ldots \to \tau_n \to \W$ and variables $X_i$ of type $\tau_i$, we will use the notation $F(X_1, \ldots,X_n)$ as a shorthand notation for $F\ (X_1)\ \ldots\ (X_n)$.

We are now ready to state the characterization of Basic Feasible Functionals in terms of moderately polynomial time functions.
\begin{thmC}[\cite{KS18}]\label{KS}
$ \lambda(\mpt)_2 = \BFF_2$.
\end{thmC}

\subsection{Proof of soundness}\label{ss:pos}

At this point we are able to give a clearer statement of the relationship between the size of a derivation for a safe program $\prog$ and the running time of a corresponding sequential execution of $\prog$. To make this precise, the running time of $M_{\phi}$ for a given input $a \in \W$ is just the number of steps
that it takes to terminate on with oracle $\phi$, starting with $a$ on its input tape (or undefined if the computation does not terminate). Given a store $\store$, this may appropriately be encoded by a single input $a_\store$. We then have
\begin{prop}%
\label{runtime}
Suppose that $\prog$ is a safe program. There are an oracle TM $M_{\phi}$ and a polynomial $P$ such that for any derivation $\pi_{\phi} : \store \Imp \prog \to \w$, $M_{\phi}$ on input $a_\store$ simulates the execution of $\prog$ with initial store $\store$ in time $O(P(\size{\pi_{\phi}}))$.
\end{prop}
\begin{proof}
By induction on the structure of the derivation of $\prog$.\qedhere
\end{proof}
An oracle $\phi'$ is padded if there exists an oracle $\phi$ such that for any binary word $w$ and integer $n$, $\phi'(w10^n)=\phi(w)$. Let $\tilde{\phi}$ denote the padded version of $\phi$.

Before showing Theorem~\ref{thm:soundness}, we first show that the complexity class $\mpt$ is invariant through padding.
\begin{prop}\label{padnpad}
$f \in\mpt$ if and only if there exists $f'  \in\mpt$ such that for any input $a$ and any oracle $\phi$, $f(\phi)(a)=f'(\tilde{\phi})(a)$.
\end{prop}

\begin{proof}
Both directions can be proved through simple rewriting of the functions.
\begin{itemize}
\item It is trivial that if $f$ on $\phi$ has finite lookahead revision $k$, replacing oracle calls $\phi(w)$ to calls to the padded oracle $\tilde{\phi}(w1)$ does not modify the lookahead revision.
\item For the other direction, assume $f'$ works on padded oracle $\tilde{\phi}$.
Then we can design $f$ working on $\phi$ such that each call to $\tilde{\phi}(w10^n)$ is replaced by 2 successive calls: $\phi(w10^n)$ and $\phi(w)$, the first being unused.
This preserves the lookahead revision.
Note that this construction may change the maximum of the lengths of the inputs and oracle answers. However, it may only increase this value, hence, if $f'$ has a polynomial step count, then so has $f$. \qedhere
\end{itemize}
\end{proof}

Now we can show a soundness result stating that any terminating and safe program computes a second order function in $\mpt$.
\begin{prop}\label{sound}
$\sem{\st} \subseteq \mpt$.
\end{prop}

\begin{proof}
Given an $\st$ program $\prog$,
by Theorem~\ref{thm:pol}, $\prog$ has a polynomial step count. By Theorem~\ref{thm:flr}, $\prog$ has finite lookahead revision and, consequently, it computes a function $f'$ in $\mpt$, however this function needs its oracles to be padded. By Proposition~\ref{padnpad}, $f'$ is equivalent to a function $f$ of $\mpt$ on standard oracles.\qedhere
\end{proof}

\begin{thm}[Soundness]\label{thm:soundness}
$\lambda(\sem{\st})_2 \subseteq \BFF_2$.
\end{thm}
\begin{proof}
By Proposition~\ref{sound}, $\sem{\st} \subseteq \mpt$.
By Lemma~\ref{closure} and Theorem~\ref{KS}, $\lambda(\sem{\st})_2 \subseteq \lambda(\mpt)_2= \BFF_2$.\qedhere
\end{proof}

Notice that the lambda-closure of $\sem{\st}$ is mandatory for characterizing $\BFF_2$ as it is known that $\mpt$ is strictly included in $\BFF_2$.
In particular, the following counter-example taken from~\cite{KS18} and computing a function of $\BFF_2$ cannot be typed as all our oracle calls have input bounds.
\begin{exa}\label{uncomp}
The functional $F$ defined below is in $\BFF_2$ but not in $\mpt$.
\begin{align*}
F(\phi,\epsilon)&=\epsilon \\
F(\phi,\msuc{1}(n))&=\phi \circ \phi (F(\phi,n) \upharpoonright \phi(\epsilon))
\end{align*}
Consequently, $F$ is not in $\sem{\st}$, since $\sem{\st} \subseteq \mpt$. Indeed, the outermost oracle call is performed without any oracle input bound and can clearly not be captured by typable programs.
\end{exa}

\section{Completenesses at type-1 and type-2}\label{s:comp}
Completeness is demonstrated in two steps. First, we show that each type 1 polynomial time computable function $\FP$ can be computed by a terminating program in $\st$, with no oracle calls. For that purpose, we show that the 2-tier sequential version of the type system of~\cite{MP14}, characterizing $\FP$, is a strict subsystem of the type system of Figure~\ref{TS}. Second, we show that the bounded  iterator functional $\mathcal{I}'$ of~\cite{KS19} can be simulated by a terminating and typable program in $\st$. The completeness follows as the type-2 simply typed lambda-closure of the bounded iterator $\mathcal{I}'$ and the functions of $\FP$ provides an alternative characterization of $\BFF_2$.

\subsection{A characterization of \texorpdfstring{$\FP$}{FP}}\label{ss:se}

For that purpose, we consider the 2-tier based characterization of $\FP$ in~\cite{MP14}, restricted to one single thread. Let $\alpha,\beta$ be tier variables ranging over $\{\tiera,\tierb\}$. The type system is provided in Figure~\ref{fig:TypeCom}, where $\overline{\alpha}$ stands for $\alpha_1 \to \ldots \to \alpha_{ar(\op)}$.

\begin{figure*}[t]
\centering
\hrulefill
\\[10pt]
\begin{prooftree}
\hypo{ \Gamma(\x)=\alpha}
\infer1[(V$_2$)]{\Gamma,\Delta \vdash_2 \x:\alpha}
\end{prooftree}
 \quad
\begin{prooftree}
\hypo{\forall i \leq n,\ \Gamma,\Delta \vdash_2 \e_i:\alpha_i \quad \overline{\alpha} \to {\alpha} \in \Delta(\op)}
\infer1[(OP$_2$)]{\Gamma,\Delta \vdash_2 \op(\e_1,\ldots,\e_{ar(\op)}):\alpha }
\end{prooftree}
\\[10pt]
 \begin{prooftree}
\hypo{\Gamma, \Delta \vdash_2 \x: \alpha  \quad \Gamma, \Delta \vdash_2 \e:\beta\quad \alpha \ord \beta}
\infer1[(A$_2$)]{\Gamma, \Delta \vdash_2 \x \asg \e:\alpha}
\end{prooftree}
\quad
\begin{prooftree}
\hypo{\phantom{sdd}}
\infer1[(SK$_2$)]{\Gamma, \Delta \vdash_2 \skp:\alpha}
\end{prooftree}
\\[10pt]
\begin{prooftree}
\hypo{\Gamma, \Delta \vdash_2 \cmd:\alpha }
\hypo{ \Gamma, \Delta \vdash_2 \cmd':\beta}
\infer2[(S$_2$)]{\Gamma, \Delta \vdash_2 \cmd \sap \cmd': \alpha \join \beta}
\end{prooftree}
\quad
\begin{prooftree}
\hypo{\Gamma, \Delta \vdash_2 \e :\alpha}
\hypo{ \Gamma, \Delta \vdash_2 \cmd: \alpha}
\hypo{\Gamma, \Delta \vdash_2 \cmd':\alpha}
\infer3[(C$_2$)]{\Gamma, \Delta \vdash_2 \ifa(\e)\{\cmd\} \elsea\{ \cmd'\}:\alpha}
\end{prooftree}
\\[10pt]
\begin{prooftree}
\hypo{\Gamma, \Delta \vdash_2 \e :\tierb}
\hypo{\Gamma, \Delta \vdash_2 \cmd:\alpha}
\infer2[(W$_2$)]{\Gamma, \Delta \vdash_2 \while (\e ) \{\cmd\}:\tierb}
\end{prooftree}
\caption{2-tier-based type system for $2\st$\label{fig:TypeCom}}
\hrulefill
\end{figure*}

In this particular context, the notion of \emph{2-tier safe program} is defined as follows. A 2-tier operator typing environment $\typop$ is a mapping that associates to each operator $\op$ a set of operator types $\typop(\op)$, of the shape $\alpha_1 \to \ldots \to \alpha_{ar(\op)} \to \alpha$, with $\alpha_i,\alpha \in \{\tiera,\tierb\}$.

\begin{defi}
A program is 2-tier safe if it can be typed using 2-tier operator typing environment $\Delta$ satisfying, for any $ \op \in \dom(\typop)$, $\sem{op} \in \FP$, $\op$ is either positive or neutral, and for each  $\alpha_1 \to \ldots \to \alpha_{ar(\op)} \to \alpha \in \typop(\op)$:
\begin{itemize}
\item $\alpha \ord \meet_{i=1,n} \alpha_i$ and
\item if $\op$ is positive but not neutral then $\alpha=\tiera$.
\end{itemize}
Let $2\st$ be the set of 2-tier safe and terminating programs and $\sem{2\st}$ be the set of functions computed by these programs.
\end{defi}

\begin{thm}[Theorem 7 of~\cite{MP14}]\label{fpmp}
$ \sem{2\st}=\FP$.
\end{thm}

We first show that the set of 2-tier safe programs is (strictly) embedded in the set of safe programs. We define a naive translation $()^\star$ from 2-tier operator typing environments to typing environments as follows:
$(\Delta)^\star$ is the operator typing environment defined by $\forall \op,\ (\Delta)^\star(\op)(\sla)=\Delta(\op)$, if $\sla = \tierb$, and $\emptyset$, otherwise.

\begin{lem}\label{t2s}
For any command  or expression $b$, if $\Gamma,\Delta \vdash_2 b : \alpha$ then $\Gamma,(\Delta)^\star \vdash b : (\alpha,\tierb,\tierb)$.
\end{lem}
\begin{proof}
By an easy induction on the typing derivation of $\Gamma,\Delta \vdash_2 b : \alpha$:
\begin{itemize}
\item If the last rule is (S$_2$), then $b = \cmd_1 \sap \cmd_2$ and $\Gamma,(\Delta)^\star \vdash \cmd_i : (\alpha_i,\tierb,\tierb)$ by the induction hypothesis, for $\alpha_i$ such that $\alpha=\alpha_1 \vee \alpha_2$. Consequently, $\Gamma,(\Delta)^\star \vdash b : (\alpha,\tierb,\tierb)$ can be derived applying rule (C) and at most once rule (SUB) (in the case where $\alpha_1 \neq \alpha_2$).
\item If the last rule is (W$_2$) then $\Gamma,\Delta \vdash_2 b : \tierb$ for $b = \while (\e ) \{\cmd'\}$ and, by the induction hypothesis, $\Gamma,(\Delta)^\star \vdash \e : (\tierb,\tierb,\tierb)$ and $\Gamma,(\Delta)^\star \vdash \cmd' : (\alpha,\tierb,\tierb)$, for some $\alpha$. Consequently, $\Gamma,(\Delta)^\star \vdash b : (\tierb,\tierb,\tierb)$ can be derived.
\item If the last  rule is (OP$_2$), then $b=\op(\overline{\e})$, for some operator $\op$ and expressions $\overline{\e}=\e_1,\ldots,\e_{ar(\op)}$ such that $\Gamma,\Delta \vdash_2 \e_i : \alpha_i$ and $\alpha_1 \to \ldots \alpha_n \to \alpha \in \Delta(\op)$. By the induction hypothesis, $\Gamma,(\Delta)^\star \vdash \e_i : (\alpha_i,\tierb,\tierb)$. Moreover, $\alpha_1 \to \ldots \alpha_n \to \alpha \in (\Delta)^\star(\op)(\tierb) = \Delta(\op)$ and, consequently, $\Gamma,(\Delta)^\star \vdash b : (\alpha,\tierb,\tierb)$ can be derived using rule (OP).
\item the cases where the last rule is (SK$_2$), (C$_2$), (V$_2$) and (A$_2$) can be straightforwardly simulated by (SK)+(SUB), (C), (V) and (A), respectively. \qedhere
\end{itemize}
\end{proof}

\begin{lem}\label{subsystem}
$2\st \subsetneq \st$.
\end{lem}

\begin{proof}
Consider a tier-2 safe program $\prog = \cmd\ \ret \x$ in $2\st$. There exist typing environments $\Gamma,\Delta$ such that $\Gamma,\Delta \vdash_2 \cmd : \alpha$. Moreover, for each  $\alpha \to \ldots \to \alpha_{n} \to \alpha \in \typop(\op)$:
\begin{itemize}
\item $\alpha \ord \meet_{i=1,n} \alpha_i$ and
\item if $\op$ is positive but not neutral then $\alpha=\tiera$.
\end{itemize}
By Lemma~\ref{t2s}, $\Gamma,(\Delta)^\star \vdash \cmd : (\alpha,\tierb,\tierb)$, and $\forall \op, (\Delta)^\star(\op)(\tierb)=\Delta(\op)$.
Consequently, for all $\alpha_1 \to \alpha_n \to \alpha \in (\Delta)^\star(\op)(\tierb)$, the following hold:
\begin{itemize}
\item $\alpha \ord \meet_{i=1,n} \alpha_i \ord \vee _{i=1,n} \alpha_i \ord \tierb$ as $\forall i,\ \alpha_i \in \{\tiera,\tierb\}$ and
 \item if $\op$ is positive but not neutral then $\alpha \ordst \tierb$, as $\alpha = \tiera$.
 \end{itemize}
 Hence $(\Delta)^\star$ is safe and $\prog$ is in $\st$.
 The inclusion is strict as $2\st$ programs have no oracle call.\qedhere
 \end{proof}

Let $\sem{\st}_1$ be defined as the set of type-1 functions computed by safe and terminating programs with no oracle calls, $\sem{\st}_1=\{\lambda w.\sem{\prog}(w) \ | \ \phi \notin \prog \text{ and } \prog \in \st \}$.

\begin{thm}\label{t1}
$\FP = \sem{\st}_1$.
\end{thm}

 \begin{proof}
By Theorem~\ref{fpmp}, any function $f$ in $\FP$ is computable by a 2-tier safe and terminating program $\prog$ with no oracle calls, \ie, $f =\lambda w.\sem{\prog}(w)$. By Lemma~\ref{subsystem}, $\prog$ is in $\st$ and, consequently, $f \in \sem{\st}_1$.
Conversely, by Corollary~\ref{psc}, if $\prog \in \st$ then $\prog$ has a polynomial step count.
Since there is no oracle call, it means that the runtime of $\prog$ is bounded by $P(\size{w})$ for some polynomial $P$.
Consequently, $\sem{\st}_1 \subseteq \FP$.\qedhere
\end{proof}

Note that the completeness part of the above Theorem ($\FP \subseteq \sem{\st}_1$) can also be proved directly by simulating polynomials over unary numbers and Turing Machines with a program in $\st$ as in the completeness proof of~\cite{M11}.

\subsection{Type two iteration}~\cite{KS19} introduces a bounded iterator functional $\mathcal{I}'$ of type $(\W \to \W) \to \W \to \W \to \W \to \W$ defined by: \[\mathcal{I}'(F,a,b,c)=(\lambda x.F(lmin(x,a)))^{\size{c}}(b),\] where $lmin$ is a functional of type $\W \to \W \to \W$ defined by: \[lmin(a,b)=
\begin{cases}
&a, \text{ if }\size{a}<\size{b}, \\
&b, \text{ otherwise.}
\end{cases}
\]

In~\cite{KS19}, using Cook's notion~\cite{C92} of polynomial time reducibility, it is shown that this functional is polynomial time-equivalent to the recursor $\mathcal{R}$ of~\cite{CU93}.
As a consequence of the Cook-Urquhart Theorem, the following characterization is obtained.

\begin{thmC}[\cite{KS19}]%
\label{th:ks19}
$ \lambda(\FP \cup \{\mathcal{I}'\})_2 = \BFF_2$.
\end{thmC}

Our proof of type-2 completeness will mostly rely on the use of this latter characterization of $\BFF_2$.

\begin{thm}[Type-2 completeness]\label{t2}
$\BFF_2 \subseteq \lambda(\sem{\st})_2$.
\end{thm}
\begin{proof}
By Theorem~\ref{t1}, $\FP	= \sem{\st}_1$, hence $\FP \subseteq \lambda(\sem{\st})_1$ by definition of lambda-closure.

Now we show that $\mathcal{I}'$ can be computed by a terminating program in $\st$. For that purpose, assume that $lmin$ is an operator of our language. $lmin$ is neutral, by definition. The program $it_\phi$, written in Figure~\ref{fig:itphi}, computes the functional $\lambda \phi.\lambda a.\lambda b.\lambda c.\mathcal{I'}(\phi,a,b,c)$ and can be typed by $(\tierb,\tierb,\tiera)$, as described in Figure~\ref{fig:iprime}, under the typing environment $\Gamma$ such that $\Gamma(c)=\Gamma(a)=\tierb$, $\Gamma(\x)=\Gamma(b)=\tiera$ and operator typing environment $\Delta$ such that $\tiera \to \tierb \to \tiera \in \Delta(lmin)(\tierb)$ and $\tierb \to \tierb \in \Delta(>0)(\tierb)$, and $\tierb \to \tierb \in \Delta(\mpred{})(\tierb)$. Note that the simulation uses the padded oracle variant $\tilde{\phi}$ of $\phi$. Recall that for each word $w$ and each integer $n$, $\tilde{\phi}(w10^n)=\phi(w)$. Consequently, the iteration of $\tilde{\phi}$ in the program $it_\phi$ simulates the iteration of $\phi$, when provided as input of $\mathcal{I}'$ (see also Proposition~\ref{padnpad}).

As $\FP \subseteq \lambda(\sem{\st})_1$ and $\mathcal{I}' \in \sem{\st}$. We have that $\lambda(\FP \cup \{\mathcal{I}'\})_2 \subseteq \lambda(\sem{\st})_2$ and the result follows by Theorem~\ref{th:ks19}.\qedhere
\end{proof}

\begin{figure}
\begin{lstlisting}
   $\x^\tiera \asg b^\tiera \sap$
   $\while (c \ne \epsilon)^\tierb\{$
      $\x^\tiera \asg \resizebox{!}{8pt}{\(\tilde{\phi}\)}(lmin(\x,a)^\tiera \upharpoonright a^\tierb)^\tiera\sap$
      $c^\tierb \asg \mpred(c)^\tierb\} $
   $\}$
   $\ret\ \x$
\end{lstlisting}
\caption{Program $it_\phi$}\label{fig:itphi}
\end{figure}

\begin{figure*}[!ht]
\hrulefill
\centering
\\[10pt]
\scalebox{0.95}{
\begin{prooftree}
\hypo{}
\ellipsis{}{}
\hypo{\Gamma(\x)=\tiera}
\infer1[(V)]{\vdash\x : (\tiera,\tierb,\tierb) }
\hypo{\Gamma(\x)=\tiera}
\infer1[(V)]{\vdash \x : (\tiera,\tierb,\tierb)}
\hypo{\Gamma(a)=\tierb}
\infer1[(V)]{\vdash a : (\tierb,\tierb,\tierb)}
\infer2[(OP)]{\vdash lmin(\x,a) : (\tiera,\tierb,\tierb)}
\hypo{\Gamma(a)=\tierb}
\infer1[(V)]{\vdash a : (\tierb,\tierb,\tierb)}
\infer2[(OR)]{\vdash \tilde{\phi}(lmin(\x,a) \upharpoonright a):(\tiera,\tierb,\tierb)}
\infer2[(A)]{\vdash\x \asg \tilde{\phi}(lmin(\x,a) \upharpoonright a):(\tiera,\tierb,\tierb) }
\infer1[(SUB)]{\vdash\x \asg \tilde{\phi}(lmin(\x,a) \upharpoonright a):(\tierb,\tierb,\tierb) }
\hypo{}
\ellipsis{}{}
\infer2[(S)]{\vdash\x \asg \tilde{\phi}(lmin(\x,a) \upharpoonright a) \sap c \asg \mpred{}(c):(\tierb,\tierb,\tierb)}
\infer2[(W$_0$)]{\vdash \while(c>0)\{\x \asg \tilde{\phi}(lmin(\x,a) \upharpoonright a) \sap c \asg \mpred{}(c)\}:(\tierb,\tierb,\tiera)}
\end{prooftree}
}
\caption{Typing derivation for program $it_\phi$}\label{fig:iprime}
\hrulefill
\end{figure*}

  To illustrate the need of the type-2 lambda closure for achieving completeness, consider a variant of Example~\ref{uncomp}:

 \begin{align*}
F'(\phi,\epsilon)&=\epsilon \\
F'(\phi,\msuc{1}(n))&=\phi \circ \phi (lmin(F'(\phi,n), \phi(\epsilon)))
\end{align*}

This functional is in $\BFF_2$ but neither in $\mpt$ nor in $\st$ as, by essence, it has no finite lookahead revision. Indeed, the outermost oracle call input data is not bounded and iterated linearly in the input. However it can be computed by $\lambda \phi.\lambda n.(\mathcal{I}'\ (\lambda x.\phi\ (\phi\ x))\ \phi(\epsilon)\ \epsilon \ n)$ and is in $\lambda(\sem{\st})_2$, as $\mathcal{I'}=\sem{it_\phi} \in \sem{\st}$, $it_\phi$ being the program in the proof of Theorem~\ref{t2}, and computes the functional $\lambda \phi.\lambda n.F'(\phi,n)$.

  \section{Other properties}\label{s:ext}
  \subsection{Intensional and extensional properties of tiers}
 The type system of Figure~\ref{TS} enjoys several other properties of interest.
 First, completeness can be achieved using only 2 tiers (at the price of worse expressive power). Second, type inference is decidable in polynomial time in the size of the program.

 Let $\sem{\st^{\sla}}$ be the subset of functionals of $\sem{\st}$ computable by terminating and typable programs using tiers bounded by $\sla$.
 Formally, $\prog \in \st^{\sla}$ if and only if $\prog$  is terminating and $\Gamma,\Delta \vdash \prog : (\sla',\sla'_{in},\sla'_{out})$ for a safe operator typing environment $\Delta$ and a variable typing environment $\Gamma$ such that $\forall \x \in \FV(\prog),\ \Gamma(\x) \ord \sla$.

 We can show that tiers allow strictly more expressive power in terms of captured programs. However tiers greater than $\tierb$ are equivalent from an extensional point of view.

\begin{prop}\label{prop:si}
The following properties hold:
\begin{enumerate}
\item $\forall \sla \succeq  \tiera,\ \st^{\sla} \subsetneq \st^{\sla+1}$,
\item $\forall \sla \succeq \tierb, \ \lambda(\sem{\st^\sla})_2 = \BFF_2$.
\end{enumerate}
\end{prop}
\begin{proof}
(1) The inclusion is trivial. For any tier $\sla+1$, it is easy to enforce the tier of one variable of a safe and terminating program to be $\sla+1$ using $\sla$ sequential while loops of the shape:
\begin{lstlisting}[]
$\while (\x_{\sla+1} > 0) \{ \x_{\sla+1} \asg \mpred(\x_{\sla+1})\sap \x_{\sla} \asg  \msuc{1}(\x_{\sla}) \}\sap$
$\cdots$
$\while (\x_{\tierb} > 0) \{ \x_{\tierb} \asg \mpred(\x_{\tierb})\sap \x_{\tiera} \asg  \msuc{1}(\x_{\tiera})\}$
 \end{lstlisting}
  Consequently, the inclusion is strict.

(2) The proof of Theorem~\ref{t2}, only makes use of programs of tier smaller than $\tierb$. Consequently, $\lambda(\sem{\st^\tierb})_2 = \BFF_2$. By Proposition~\ref{prop:si}, $\st^{\sla} \subsetneq \st^{\sla+1}$ and, consequently, $\forall \sla,\ \sem{\st^{\sla}} \subseteq \sem{\st^{\sla+1}}$. We obtain that $\forall \sla,\ \BFF_2 \subseteq \lambda(\sem{\st^{\sla}})_2 \subseteq  \lambda(\sem{\st})_2 = \BFF_2$ and so the result.\qedhere
\end{proof}

Proposition~\ref{prop:si} implies that the use of exactly 2 tiers is sufficient to achieve completeness but weakens the type system expressive power.

\subsection{Decidability of type inference}

\begin{prop}\label{prop:ti}
Given a program $\prog$ of size $n$ and a safe operator typing environment $\Delta$, deciding if there exists a variable typing environment $\Gamma$ such that $\prog \in \st^\sla$ can be done in time $\mathcal{O}(n^2\times \sla)$.
\end{prop}

\begin{proof}
The proof follows the type inference proof of~\cite{HMP13}: the tier of each variable $\x$ is encoded using $3(\sla+1)$ boolean variables $x_\tiera^{tier}$, $x_\tierb^{tier}$, $\ldots$, $x_\sla^{tier}$, $x_\tiera^{in}$, $x_\tierb^{in}$, $\ldots$, $x_\sla^{in}$,$x_\tiera^{out}$, $x_\tierb^{out}$, $\ldots$, $x_\sla^{out}$. The same encoding can be done for an expression $\e$ and a command $\cmd$ using the boolean variables $e$ and $c$. These variables can be interpreted as follows: $x_\sla^{in}$ is true whenever the innermost tier of variable $\x$ is at most $\sla$ and $c_\tierb^{tier}$ is false whenever the tier of command $\cmd$ is strictly more than $\tierb$. For example, if the tier of $\x$ is $\tierb$, then $x_{\tiera}^{tier}$ is false and $x_{\tierb}^{tier},\ldots,  x_{\sla}^{tier}$ are true.

The tiers of each boolean variable $a \in \{x,e,c\}$ is enforced to be correctly encoded by the following propositional formula $\bigwedge_{k \in \{tier,in,out\}}\bigwedge_{i \ordst \sla} (a_{i}^k \implies a_{i+1}^k)$, which is equivalent to $\bigwedge_{k \in \{tier,in,out\}}\bigwedge_{i \ordst \sla}(\neg {a_{i}^k}\vee {a_{i+1}}^k)$.
This accounts for $3\sla$ clauses for each variable.


Equality of the $k$-tier of variables $a$ and the $l$-tier of variable $b$ (with $k,l \in \{tier,in, out\}$) can be expressed as:
\[\bigwedge_{i}(\neg {a_{i}^k}\vee {b_{i}}^{l}) \wedge \bigwedge_{i}({a_{i}^k}\vee  \neg{b_{i}}^{l}).\]
This accounts for $2 (\sla+1)$ clauses.

Strict inequality of tiers, for example the tier of $\e$ strictly less than the innermost tier of $\mathtt{d}$, can be encoded as:
\[\bigwedge_{i\ordst \sla} (\neg {d}^{in}_{i+1} \vee {e}^{tier}_i) \wedge \neg d^{in}_0 \wedge {e}^{tier}_{\sla}.\]
This accounts for $\sla+2$ clauses.

Open inequality of tiers, for example the tier of variable $\x$ is at most the tier of variable $\y$, can be encoded as:
\[\bigwedge_i (\neg {y}^{tier}_i \vee {x}^{tier}_{i}).\]
This accounts for $\sla+1$ clauses.

Now we inspect each of the program constructs relatively to the corresponding typing rule in Figure~\ref{TS}:
\begin{itemize}
\item Consider typing rule $(A)$.

\[\begin{prooftree}
\hypo{\pbl \x: (\sla_1,\sla_{in},\sla_{out})}
\hypo{\pbl \e: (\sla_2,\sla_{in},\sla_{out})}
\hypo{\sla_1  \ord \sla_2}
\infer3[(A)]{\pbl \x \asg \e \ : (\sla_1,\sla_{in},\sla_{out})}
\end{prooftree}.\]
 The typing of command $\cmd \triangleq \x \asg \e$ translates to 1 open inequality and 5 equalities, accounting for $11(\sla+1)$ clauses:

\[
\begin{array}{l}
\bigwedge_i(\neg e^{tier}_i \vee x^{tier}_i) \wedge\\
\bigwedge_i(\neg c_i^{tier}\vee x_i^{tier})\wedge \bigwedge_i(c_i^{tier}\vee \neg x_i^{tier})\wedge\\
\bigwedge_i(\neg c_i^{in}\vee x_i^{in})    \wedge \bigwedge_i(c_i^{in}  \vee \neg x_i^{in})  \wedge\\
\bigwedge_i(\neg c_i^{in}\vee e_i^{in})    \wedge \bigwedge_i(c_i^{in}  \vee \neg e_i^{in})  \wedge\\
\bigwedge_i(\neg c_i^{out}\vee x_i^{out})  \wedge \bigwedge_i(c_i^{out} \vee \neg x_i^{out}) \wedge\\
\bigwedge_i(\neg c_i^{out}\vee e_i^{out})  \wedge \bigwedge_i(c_i^{out} \vee \neg e_i^{out})
\end{array}
\]
\item Similarly, consider typing rule $(W)$:
\[
\begin{prooftree}
\hypo{\pbl \e: (\sla_1,\sla_{in},\sla_{out})}
\hypo{\pbl \cmd : (\sla_1,\sla_1,\sla_{out})}
\hypo{\tierb \ord \sla_1 \ord \sla_{out}}
\infer3[(W)]{\pbl  \while (\e) \{\cmd\}\ :(\sla_1,\sla_{in},\sla_{out})}
\end{prooftree}.
\]
The typing of  command $\mathtt{w} \triangleq \while(\e)\{\cmd\}$ translates to 2 inequalities ($\tierb\ord\sla_1$ which is trivial and  $\sla_1 \ord \sla_{out}$) and 6 equalities, accounting for $13(\sla+1)+1$ clauses:

\[
\begin{array}{l}
\neg w_0^{tier} \wedge \\
\bigwedge_i \neg w_i^{out} \vee w_i^{tier} \wedge\\
\bigwedge_i \neg e^{tier}_i \vee w^{tier}_i \wedge \bigwedge_i e^{tier}_i \vee \neg w^{tier}_i\wedge\\
\bigwedge_i \neg c^{tier}_i \vee w^{tier}_i \wedge \bigwedge_i c^{tier}_i \vee \neg w^{tier}_i\wedge\\
\bigwedge_i \neg c^{in}_i \vee w^{tier}_i \wedge \bigwedge_i c^{in}_i \vee \neg w^{tier}_i\wedge\\
\bigwedge_i \neg e^{in}_i \vee w^{in}_i \wedge \bigwedge_i e^{in}_i \vee \neg w^{in}_i\wedge\\
\bigwedge_i \neg c^{out}_i \vee w^{out}_i \wedge \bigwedge_i c^{out}_i \vee \neg w^{out}_i\wedge\\
\bigwedge_i \neg e^{out}_i \vee w^{out}_i \wedge \bigwedge_i e^{out}_i \vee \neg w^{out}_i\wedge\\
\end{array}\]
\item The number of equalities used when encoding a rule $(OP)$ is $3\times ar(op)+1$, which needs in total $(3\times ar(op)+1)\times(2(\sla+1))$ clauses. Hence $\mathcal{O}(n\times\sla)$ clauses.
\end{itemize}

\noindent
Finally, the type inference problem can be reduced to 2-SAT with $\mathcal{O}(n^2\times \sla)$ clauses, which can be solved in time linear in the number of clauses~\cite{EveItaSha76,AspPlaTar79}.\qedhere
\end{proof}

\begin{thm}\label{thm:ti}
Given a program $\prog$ and a safe operator typing environment $\Delta$, deciding if there exists a variable typing environment $\Gamma$ such that $\prog \in \st$ can be done in time cubic in the size of the program.
\end{thm}

\begin{proof}
The maximal tier needed to type a program can be bounded by the size of the program as the number of strict inequalities on tiers is fixed by the number of rules (OP) (in the case of a positive operator) and (OR) needed to type a program. Consequently, with $n$ the size of $\prog$, we can simply check if $\prog \in \st^{n}$.
By Proposition~\ref{prop:ti}, this means that we can decide if $\prog \in \st$ in time $\mathcal{O}(n^3)$.\qedhere
\end{proof}

\section{Conclusion and future work}\label{s:con}
We have presented a first tractable characterization of the class of type-2 polynomial time computable functionals $\BFF_2$ based on a simple imperative programming language. This characterization does not require any explicit and external resource bound and its restriction to type-1 provides an alternative characterization of the class $\FP$.

The presented type system can be generalized to programs with a constant number of oracles (the typing rule for oracles remains unchanged).
However the lambda closure is mandatory for completeness as illustrated by Example~\ref{uncomp}. An open issue of interest is to get rid of this closure in order to obtain a characterization of $\BFF_2$ in terms of a pure imperative programming language. Indeed, in our context, programs can be viewed as a simply typed lambda-terms with typable and terminating imperative procedure calls. One suggestion is to study to which extent oracle composition can be added directly to the program syntax.

Another issue of interest is to study whether this non-interference based approach could be extended (or adapted within the context of light logics) to characterize $\BFF_2$ on a pure functional language. We leave these open issues as future work.

\section*{Acknowledgements.} We would like to thank the anonymous reviewers for their suggestions and comments, which helped us to greatly improve the presentation of our work. Bruce Kapron's work was supported in part by NSERC RGPIN-2021-02481. 

\bibliographystyle{alphaurl}

\begin{thebibliography}{HKMP20}

\bibitem[APT79]{AspPlaTar79}
Bengt Aspvall, Michael~F. Plass, and Robert~Endre Tarjan.
\newblock A linear-time algorithm for testing the truth of certain quantified
  boolean formulas.
\newblock {\em Information Processing Letters}, 8(3):121--123, 1979.

\bibitem[BAJK08]{BAJK08}
Amir~M. Ben-Amram, Neil~D. Jones, and Lars Kristiansen.
\newblock Linear, polynomial or exponential? complexity inference in polynomial
  time.
\newblock In {\em Logic and Theory of Algorithms}, pages 67--76. Springer,
  2008.

\bibitem[BC92]{BelCoo92}
Stephen Bellantoni and Stephen Cook.
\newblock A new recursion-theoretic characterization of the polytime functions.
\newblock {\em Computational Complexity}, 2:97--110, 1992.

\bibitem[BL16]{BL16}
Patrick Baillot and Ugo~Dal Lago.
\newblock Higher-order interpretations and program complexity.
\newblock {\em Information and Computation}, 248:56--81, 2016.

\bibitem[BM10]{BM10}
Patrick Baillot and Damiano Mazza.
\newblock Linear logic by levels and bounded time complexity.
\newblock {\em Theoretical Computer Science}, 411(2):470--503, 2010.

\bibitem[BMM11]{BMM11}
Guillaume Bonfante, Jean{-}Yves Marion, and Jean{-}Yves Moyen.
\newblock Quasi-interpretations a way to control resources.
\newblock {\em Theoretical Computer Science}, 412(25):2776--2796, 2011.

\bibitem[BT04]{BT04}
Patrick Baillot and Kazushige Terui.
\newblock Light types for polynomial time computation in lambda-calculus.
\newblock In {\em Logic in Computer Science, {LICS} 2004}, pages 266--275. IEEE,
  2004.

\bibitem[CK89]{CooKap89}
Stephen~A. Cook and Bruce~M. Kapron.
\newblock Characterizations of the basic feasible functionals of finite type.
\newblock In {\em Symposium on Foundations of Computer Science,
  {FOCS} 1989}, pages 154--159. IEEE, 1989.

\bibitem[Cob65]{Cob65}
Alan Cobham.
\newblock The intrinsic computational difficulty of functions.
\newblock In {\em International
  Conference on Logic, Methodology, and Philosophy of Science}, pages 24--30.
  North-Holland, Amsterdam, 1965.

\bibitem[Con73]{Con73}
Robert~L. Constable.
\newblock Type two computational complexity.
\newblock In {\em Symposium on Theory of Computing, {STOC} 1973}, pages
  108--121. ACM, 1973.

\bibitem[Coo92]{C92}
Stephen~A. Cook.
\newblock Computability and complexity of higher type functions.
\newblock In {\em Logic from Computer Science}, pages 51--72. Springer, 1992.

\bibitem[CPR06]{CPR06}
Byron Cook, Andreas Podelski, and Andrey Rybalchenko.
\newblock Terminator: beyond safety.
\newblock In {\em International Conference on Computer Aided Verification, {CAV} 2006},
  pages 415--418. Springer, 2006.

\bibitem[CU93]{CU93}
Stephen~A. Cook and Alasdair Urquhart.
\newblock Functional interpretations of feasibly constructive arithmetic.
\newblock {\em Annals of Pure and Applied Logic}, 63(2):103--200, 1993.

\bibitem[DR06]{DR06}
Norman Danner and James~S. Royer.
\newblock Adventures in time and space.
\newblock In {\em Symposium on
  Principles of Programming Languages, {POPL} 2006}, pages 168--179. ACM, 2006.

\bibitem[EIS76]{EveItaSha76}
Shimon Even, Alon Itai, and Adi Shamir.
\newblock On the complexity of timetable and multicommodity flow problems.
\newblock {\em SIAM Journal on Computing}, 5(4):691--703, 1976.

\bibitem[FHHP15]{FHHP15}
Hugo F{\'{e}}r{\'{e}}e, Emmanuel Hainry, Mathieu Hoyrup, and Romain
  P{\'{e}}choux.
\newblock Characterizing polynomial time complexity of stream programs using
  interpretations.
\newblock {\em Theoretical Computer Science}, 585:41--54, 2015.

\bibitem[Gir98]{Girard98}
Jean-Yves Girard.
\newblock Light linear logic.
\newblock {\em Information and Computation}, 143(2):175--204, 1998.

\bibitem[GMR08]{GMR08}
Marco Gaboardi, Jean{-}Yves Marion, and Simona Ronchi~Della Rocca.
\newblock A logical account of {PSPACE}.
\newblock In {\em Symposium on Principles of Programming Languages, {POPL} 2008}, pages 121--131. ACM, 2008.

\bibitem[H{\'{a}}j79]{H79}
Petr H{\'{a}}jek.
\newblock Arithmetical hierarchy and complexity of computation.
\newblock {\em Theoretical Computer Science}, 8:227--237, 1979.

\bibitem[HKMP20]{HKMP20}
Emmanuel Hainry, Bruce~M. Kapron, Jean{-}Yves Marion, and Romain P{\'{e}}choux.
\newblock A tier-based typed programming language characterizing feasible
  functionals.
\newblock In {\em Symposium on Logic in
  Computer Science, {LICS} 2020}, pages
  535--549, 2020.

\bibitem[HMP13]{HMP13}
Emmanuel Hainry, Jean-Yves Marion, and Romain P{\'e}choux.
\newblock Type-based complexity analysis for fork processes.
\newblock In {\em International Conference on Foundations of Software Science
  and Computational Structures, {FoSSaCS} 2013}, pages 305--320. Springer, 2013.

\bibitem[HP15]{HP15}
Emmanuel Hainry and Romain P{\'{e}}choux.
\newblock Objects in polynomial time.
\newblock In {\em Asian Symposium on Programming Languages and Systems,
  {APLAS} 2015}, Lecture Notes in Computer Science, pages 387--404. Springer,
  2015.

\bibitem[HP17]{HP17}
Emmanuel Hainry and Romain P{\'{e}}choux.
\newblock Higher order interpretation for higher order complexity.
\newblock In {\em International Conference on Logic for
  Programming, Artificial Intelligence and Reasoning, {LPAR} 2017}, pages 269--285.
  EasyChair, 2017.

\bibitem[IRK01]{IRK01}
Robert~J. Irwin, James~S. Royer, and Bruce~M. Kapron.
\newblock On characterizations of the basic feasible functionals (part {I}).
\newblock {\em Journal of Functional Programming}, 11(1):117--153, 2001.

\bibitem[JK09]{JK09}
Neil~D. Jones and Lars Kristiansen.
\newblock A flow calculus of mwp-bounds for complexity analysis.
\newblock {\em ACM Transactions on  Computational Logic}, 10(4):28:1--28:41, 2009.

\bibitem[KC91]{KC91}
Bruce~M. Kapron and Stephen~A. Cook.
\newblock A new characterization of {M}ehlhorn's polynomial time functionals
  (extended abstract).
\newblock In {\em Symposium on Foundations of Computer Science,  {FOCS} 1991}, pages 342--347. IEEE, 1991.

\bibitem[KC96]{KC96}
Bruce~M. Kapron and Stephen~A. Cook.
\newblock A new characterization of type-2 feasibility.
\newblock {\em SIAM Journal on Computing}, 25(1):117--132, 1996.

\bibitem[KS17]{KS17}
Akitoshi Kawamura and Florian Steinberg.
\newblock Polynomial running times for polynomial-time oracle machines.
\newblock In {\em International Conference on Formal Structures for
  Computation and Deduction, {FSCD} 2017}, pages 23:1--23:18. Schloss Dagstuhl
  - Leibniz-Zentrum fuer Informatik, 2017. 

\bibitem[KS18]{KS18}
Bruce~M. Kapron and Florian Steinberg.
\newblock Type-two polynomial-time and restricted lookahead.
\newblock In {\em Logic in Computer Science, {LICS} 2018}, pages 579--588. ACM,
  2018.

\bibitem[KS19]{KS19}
Bruce~M. Kapron and Florian Steinberg.
\newblock Type-two iteration with bounded query revision.
\newblock In {\em Joint Workshops on Developments in Implicit
  Computational complExity and Foundational {\&} Practical Aspects of Resource
  Analysis, {DICE-FOPARA@ETAPS} 2019}, {EPTCS}, pages 61--73, 2019.

\bibitem[Lei95]{L94}
Daniel Leivant.
\newblock Ramified recurrence and computational complexity {I}: Word recurrence
  and poly-time.
\newblock In {\em Feasible
  Mathematics II}, pages 320--343. Birkh{\"a}user, Boston, MA, 1995.

\bibitem[LJB01]{LJB01}
Chin~Soon Lee, Neil~D. Jones, and Amir~M. Ben{-}Amram.
\newblock The size-change principle for program termination.
\newblock In {\em Symposium on Principles of Programming Languages, {POPL} 2001}, pages
  81--92. ACM, 2001.

\bibitem[LM93]{LeivantMar93}
Daniel Leivant and Jean-Yves Marion.
\newblock Lambda calculus characterizations of poly-time.
\newblock {\em Fundamenta Informaticae}, 19(1/2):167--184, 1993.

\bibitem[LM13]{LM13}
Daniel Leivant and Jean{-}Yves Marion.
\newblock Evolving graph-structures and their implicit computational
  complexity.
\newblock In {\em International Colloquium on Automata, Languages, and Programming, {ICALP} 2013, Part {II}}, Lecture Notes in Computer
  Science, pages 349--360. Springer, 2013.

\bibitem[Mar11]{M11}
Jean{-}Yves Marion.
\newblock A type system for complexity flow analysis.
\newblock In {\em Logic in Computer Science, {LICS} 2011}, pages 123--132. IEEE
  Computer Society, 2011.

\bibitem[Meh76]{M76}
Kurt Mehlhorn.
\newblock Polynomial and abstract subrecursive classes.
\newblock {\em Journal of Computer and System Sciences}, 12(2):147--178, 1976.

\bibitem[Mit91]{M91}
John~C. Mitchell.
\newblock Type inference with simple subtypes.
\newblock {\em Journal of Functional Programming}, 1(3):245--285, 1991.

\bibitem[MP14]{MP14}
Jean{-}Yves Marion and Romain P{\'{e}}choux.
\newblock Complexity information flow in a multi-threaded imperative language.
\newblock In {\em Theory and Applications of Models of Computation, {TAMC}
  2014}, Lecture Notes in Computer Science, pages 124--140. Springer, 2014.

\bibitem[VIS96]{VIS96}
Dennis Volpano, Cynthia Irvine, and Geoffrey Smith.
\newblock A sound type system for secure flow analysis.
\newblock {\em Journal of Computer Security}, 4(2-3):167--187, 1996. 

\end{thebibliography}

\end{document}